\documentclass[letterpaper,journal]{IEEEtran}
\usepackage{amsmath,amsfonts}
\usepackage{algorithm}
\usepackage{array}
\usepackage[caption=false,font=normalsize,labelfont=sf,textfont=sf]{subfig}
\usepackage{textcomp}
\usepackage{stfloats}
\usepackage{xurl}
\usepackage[hidelinks]{hyperref}
\usepackage{verbatim}
\usepackage{multirow}
\usepackage{graphicx}
\usepackage{bbm}
\usepackage{cite}
\usepackage{vmr-symbols-vecbold}
\usepackage{standard-macros}
\usepackage{algorithm}
\usepackage{algpseudocode}
\usepackage{cuted}
\usepackage{mathdots}
\usepackage{booktabs}
\usepackage{multirow}
\usepackage{siunitx}

\usepackage{tikz,pgfplots}
\usepackage{pgfplotstable}
\usetikzlibrary{positioning}
\usetikzlibrary{decorations.pathreplacing}
\usetikzlibrary {arrows.meta,bending,positioning}
\pgfplotsset{compat=newest}
\usepgfplotslibrary{patchplots}
\usepgfplotslibrary{statistics}
\usetikzlibrary{backgrounds,plotmarks,fit,positioning,arrows,shapes,shapes.multipart,calc,arrows.meta}
\usetikzlibrary{fit,positioning,arrows,shapes,shapes.multipart,calc,arrows.meta,shapes.geometric,shapes.misc}
\usetikzlibrary{decorations.pathreplacing,angles,quotes,calligraphy}
\usetikzlibrary{automata}
\usetikzlibrary{arrows.meta, positioning, backgrounds, fit}

\usepackage{amsthm}
\newtheorem{theorem}{Theorem}

\DeclareMathOperator{\Bin}{Bin}

\DeclareMathOperator{\Unif}{Unif}
\DeclareMathOperator{\enc}{enc}
\DeclareMathOperator{\dec}{dec}
\DeclareMathOperator{\quant}{quant}
\DeclareMathOperator{\sens}{sens}
\newcommand{\adjustedbar}[1]{\overline{\hspace{-0.1em}#1\hspace{-0.1em}}}
\newcommand{\dTV}{\ensuremath{\overline{\opT\opV}}}

\makeatletter
\newcommand\fs@betterruled{
  \def\@fs@cfont{\bfseries}\let\@fs@capt\floatc@ruled
  \def\@fs@pre{\vspace*{6pt}\hrule height.8pt depth0pt \kern2pt}
  \def\@fs@post{\kern2pt\hrule\relax}
  \def\@fs@mid{\kern2pt\hrule\kern2pt}
  \let\@fs@iftopcapt\iftrue}
\floatstyle{betterruled}
\restylefloat{algorithm}
\makeatother

\begin{document}

\title{Type-Based Unsourced Multiple Access \\Over Fading Channels in Distributed MIMO  \\With Application to Multi-Target Localization}

\author{
Kaan Okumus, \textit{Graduate Student Member}, \textit{IEEE}, Khac-Hoang Ngo, \textit{Member}, \textit{IEEE}, \\ Giuseppe Durisi, \textit{Senior Member}, \textit{IEEE}, Erik G. Str\"om, \textit{Fellow}, \textit{IEEE}

\thanks{Kaan Okumus, Giuseppe Durisi, and Erik G. Str\"om are with the Department of Electrical Engineering, Chalmers University of Technology, 41296 Gothenburg, Sweden (email: \{okumus, durisi, erik.strom\}@chalmers.se).}
\thanks{Khac-Hoang Ngo is with the Department of Electrical Engineering, Link\"oping University, 58183 Link\"oping, Sweden (email: khac-hoang.ngo@liu.se).}
\thanks{This work was supported in part by the Swedish Research Council under grants 2021-04970 and 2022-04471, and by the Swedish Foundation for Strategic Research. The work of K.-H. Ngo was supported in part by the Excellence Center at Linköping – Lund in Information Technology (ELLIIT). This paper was presented in part at the IEEE International Workshop on Signal Processing Advances in Wireless Communications (SPAWC), Lucca, Italy, 2024 \cite{ngo2024_tuma} and at the IEEE International Symposium on Information Theory (ISIT), Ann Arbor, Michigan, USA, 2025 \cite{tuma_fading_2025}.}
}

\maketitle

\begin{abstract}
We consider the problem of type estimation over unsourced multiple access fading channels in distributed multiple-input multiple-output (D-MIMO) systems. Unlike classical unsourced multiple access, type-based unsourced multiple access (TUMA) aims to estimate the type, i.e., the empirical distribution of transmitted messages. We extend our prior work on TUMA over additive white Gaussian channels to fading scenarios in which neither the transmitters nor the receiver have channel state information. To mitigate the impact of path-loss variability, we employ location-based codebook partitioning: users with similar large-scale fading coefficients use the same codebook. 
The decoder is built on the multisource approximate message passing algorithm proposed by Çakmak \textit{et al.} (2025), and supports both centralized and distributed implementations. 
As an application, we demonstrate how TUMA enables efficient communication in a multi-target localization setting, where distributed sensors report to a D-MIMO receiver quantized target positions. 
We propose a performance cost function that combines localization errors with a misdetection penalty, and use it to characterize how performance depends on the fraction of resources assigned to sensing vs. communication, as well as on the number of bits used to quantize the positions of the targets. 

\end{abstract}


\section{Introduction}

\IEEEPARstart{M}{assive} machine-type communication (mMTC) is a key enabler for emerging 6G applications: it supports simultaneous transmissions from a large number of devices over a shared wireless medium. These devices typically transmit short packets in a sporadic and uncoordinated manner and operate under strict energy constraints~\cite{chen_mmtc, wu_mmtc, kalor_giuseppe_6G}. With the need to support massive device connectivity, any scheme that assigns dedicated resources such as time, frequency, or codes to each device is not scalable. To address this challenge, the unsourced multiple access (UMA) framework was introduced~\cite{polyanskiy2017}. In UMA, all users share a common codebook and the receiver is tasked with recovering the list of transmitted messages without identifying their sources. Recent significant theoretical and algorithmic advances in this area include the analysis of fundamental limits~\cite{lancho_uma_a_channel, gao_ml_bound, andreev_uma_projection_bound, ngo_2023_random_user_activity, frolov_monograph} on progressively more realistic channel models, and the design of coding schemes approaching these limits~\cite{frolov_monograph, fengler_sparc, amalladinne_enhanced_ccsamp, pradhan_graph_uma_decoder, vem_successive_decoder, fengler_nonbayesian, han_sparse_kronecker_decoder, decurninge_tensor_based_uma, pradhan_polar_code_uma}. A key assumption in UMA is that the event where multiple devices transmit the same message simultaneously, known as message collision, is treated as a decoding error. This is justified because, when each device selects its message independently and uniformly at random from a large message set, the probability of message collisions is negligible. However, this assumption is often violated in practical sensing and learning applications, where multiple users may observe the same phenomenon and generate identical messages. Examples include multi-target localization~\cite{Hoffman2004, Schuhmacher2008, Rahmathullah2017}, where sensors may report the same target position; federated learning~\cite{qiao_gunduz_fl, okumus2026tuma_fl}, 
where clients may transmit the same quantized model updates due to similarities in their local data; and point-cloud transmission~\cite{Bian24wpc}, where devices observing the same surface may produce identical descriptors. In such cases, message collisions occur naturally and can no longer be neglected. 
Conventional UMA decoders can identify only whether a message is present or absent. Consequently, when multiple users transmit the same message, their contributions are merged into a single detected message. As a result, valuable information about the underlying phenomenon may be lost. In the abovementioned applications, however, it is critical that the decoder provides an estimate of the multiplicity of each message. This motivates the need for a generalized UMA framework, in which the multiplicity of each message is reliably estimated. This paper addresses this gap by proposing type-based UMA (TUMA), which handles the estimation of message multiplicities over fading channels in a distributed multiple-input multiple-output (D-MIMO) architecture.

\begin{table*}[t]
\centering
\caption{Comparison Between This Work and Closely Related Prior Works}
\label{tab:comparison}
\renewcommand{\arraystretch}{1.15}

\begin{tabular}{c||c|c|c|c|c|c}
\toprule
Ref. &
Fading &
CSI required &
D-MIMO &
Collision-aware (TUMA) &
Decoder &
Application \\
\midrule

\cite{polyanskiy2017} &
no &
$-$ &
no &
no &
maximum likelihood &
$-$ \\

\cite{cakmak_2025_journal} &
yes &
no &
yes &
no &
multisource AMP &
$-$ \\

\cite{qiao_gunduz_fl} &
yes &
yes &
no &
yes &
AMP-DA & 
federated learning \\

$\,$\cite{ngo2024_tuma}$^\ast$ &
no &
$-$ &
no &
yes &
AMP &
multi-target localization \\

$\,$\cite{tuma_fading_2025}$^\ast$ &
yes &
no &
yes &
yes &
multisource AMP &
$-$ \\

\cite{okumus2026tuma_fl} &
yes &
no &
yes &
yes &
multisource AMP & 
federated learning \\

\midrule
this work &
yes &
no &
yes &
yes &
multisource AMP & 
multi-target localization \\

\bottomrule
\end{tabular}

\vspace{1mm}
\begin{flushleft}
\footnotesize{
\hspace{1.2cm} $^\ast$ conference papers that make up parts of this paper\\
\hspace{1.2cm} $-$ indicates that the feature is not supported or not considered
\vspace{-5mm}
}
\end{flushleft}
\end{table*}

\subsection{Related Works}

Estimating the multiplicity of each message is equivalent to characterizing the empirical distribution (also known as ``type") of the messages. An early formulation of type-based estimation appeared in the context of distributed detection in wireless sensor networks~\cite{liu_asymp_type_2004}, where sensors transmit local type summaries to a fusion center. This idea was later formalized under the name of type-based multiple access in~\cite{Mergen2006}, where message multiplicities are estimated via the use of orthogonal codebooks. This approach, however, suffers from limited scalability, due to the high spectral efficiency cost of orthogonality. To overcome this limitation, recent works extended type-based multiple access to nonorthogonal codebooks~\cite{dommel_jscbmp_2020, dommel_fog_2021, zhu_IB_tbma_2023}. 
However, these works consider only a small message set and do not address the massive and sporadic device activity that is typical in massive random access. A type estimation scheme with nonorthogonal codebooks has also been presented in~\cite{qiao_gunduz_fl} in the context of over-the-air computation for federated learning, but the proposed approximate message passing (AMP)-based digital aggregation (AMP-DA) decoder requires perfect channel state information (CSI) at the transmitter side to perform channel pre-equalization. A general type-based access framework supporting nonorthogonal transmissions, a large message space and a massive number of sporadically activated devices, and operating under practical constraints, such as the lack of CSI, is not available in the literature.

To address these limitations, in our previous work~\cite{ngo2024_tuma}, we introduced the TUMA framework---a general communication-theoretic solution for estimating message types under collisions. TUMA supports massive uncoordinated access with nonorthogonal codewords and enables type estimation via Bayesian inference. In~\cite{ngo2024_tuma}, three decoder architectures were considered: two based on expectation propagation techniques adopted from~\cite{Meng_NOMA_2021}, and one based on the AMP decoder from~\cite{fengler_sparc}. We showed that the AMP-based decoder achieves superior performance in terms of complexity and estimation accuracy. In~\cite{deekshith_tuma_bound}, we also derived an achievability bound for TUMA and showed a good agreement between the bound and the performance achievable using a practical coded compressed sensing (CCS)-AMP scheme~\cite{amalladine_ccs_amp}, adapted for TUMA. The analysis in~\cite{ngo2024_tuma, deekshith_tuma_bound} is, however, conducted for a simplified Gaussian channel model, in which small-scale fading is absent and all users are received at the same power. 

Such limitations have recently been overcome in~\cite{gkiouzepi2024jointmessagedetectionchannel, cakmak_2025_journal, cakmak_ISIT25}, where the authors analyze UMA over D-MIMO network architectures. The analysis in~\cite{gkiouzepi2024jointmessagedetectionchannel, cakmak_2025_journal, cakmak_ISIT25} relies on two key ideas: i) the assumption that an accurate model for the large-scale fading coefficients (LSFCs) at each point in the coverage area is available at the receiver; ii) the use of location-based codebook partitioning, where each user selects its codebook based on its LSFC profile. For this setup, 
in~\cite{gkiouzepi2024jointmessagedetectionchannel, cakmak_2025_journal, cakmak_ISIT25}, a decoding algorithm, referred to as multisource AMP, is proposed, which allows for the estimation of the transmitted messages, the channel statistics, and the user location without requiring any knowledge of the instantenous CSI. The use of a location-based codebook partitioning is advantageous, because it allows one to use a more concentrated prior for the path-loss vector at the denoiser, which is beneficial in terms of performance. However, the analysis in~\cite{gkiouzepi2024jointmessagedetectionchannel, cakmak_2025_journal, cakmak_ISIT25} does not address type estimation. Indeed, in~\cite{gkiouzepi2024jointmessagedetectionchannel,
cakmak_2025_journal, cakmak_ISIT25}, the Bayesian denoisers
are designed for joint message detection and channel estimation, under the assumption that each transmitted codeword is associated with a single user. In contrast, TUMA requires estimating the combined contribution of multiple users transmitting the same codeword, which necessitates a different prior model, denoiser design, and multiplicity estimation procedure.

\subsection{Contributions}

In this work, we extend our preliminary conference works~\cite{ngo2024_tuma,tuma_fading_2025} on TUMA framework to support a more realistic system model and evaluate it in a multi-target localization application. Table~\ref{tab:comparison} provides a high-level overview of the key differences between this work and the most closely related prior works. Our specific contributions are as follows:

\paragraph*{TUMA over fading channels in D-MIMO with random user activity}
We generalize the TUMA framework in~\cite{tuma_fading_2025} by including both small- and large- scale fading and by assuming that the number of active users is random and unknown to the receiver. As in~\cite{cakmak_2025_journal}, we assume perfect knowledge of LSFC at each point in the coverage area and employ location-based codebook partitioning. We adapt the multisource AMP algorithm proposed in~\cite{cakmak_2025_journal} to TUMA by modifying the Bayesian denoiser to account for message multiplicities. In multisource AMP, the denoiser estimates, for each codeword, the combined fading coefficients of all users who transmitted that codeword. Under message collisions, this requires integrating over multiple unknown user positions, which leads to high complexity. We employ Monte Carlo (MC) sampling to approximate these integrals. We also modify the Onsager correction term, which compensates for correlation across AMP iterations, to ensure consistency with the modified denoiser. Based on the AMP outputs, we perform Bayesian estimation of the message multiplicities and normalize them to obtain the estimated type. We additionally introduce a distributed variant of our decoder, inspired by the distributed AMP (dAMP) algorithm~\cite{bai_larsson_damp}, originally  proposed in the context of activity detection under perfect per user LSFC information. The distributed variant leads to a lower computational complexity, which makes it suitable for large-scale D-MIMO systems.

\paragraph*{Application to multi-target localization and end-to-end tradeoff evaluation}
We demonstrate how TUMA can be applied to multi-target localization, where distributed sensors detect targets and transmit their quantized positions via unsourced access. Localization problems in D-MIMO networks have also been studied via integrated sensing and communication approaches, where targets are localized through the reflected communication paths they induce~\cite{gao_isac_comm_25}. In contrast, we adopt a sensor-assisted architecture in which sensing and communication are performed separately, and we focus on providing reliable communication of sensed target information under message collisions. This architectural choice is consistent with the transparent sensing paradigm defined and studied in~\cite[Sec.~5.4]{3gpp22837}, in which sensing measurements are first acquired by distributed sensors and subsequently communicated through digital links. Specifically,
the sensing model used in our system is based on a practical millimeter-wave (mmWave) 5G simultaneous localization and mapping (SLAM) scenario, with detection probabilities depending on the path-loss~\cite{henk_slam_pdet_2020}. Communication instead is performed over a lower frequency narrowband channel. For this setting, we introduce a prior on the message multiplicities that reflects the sensing model. To assess performance, we fix a total transmission blocklength and split it into a phase devoted to target sensing and one devoted to communication of the estimated target positions. A larger sensing blocklength improves detection but reduces the blocklength available for communication, which degrades type estimation. A large communication blocklength has the opposite effect. The quantization resolution adds a third component to this tradeoff: using more bits reduces localization error but enlarges the message space, making type estimation more challenging.
    
To evaluate these effects, we employ complementary metrics. We use total variation (TV) distance to measure communication errors, and Wasserstein distance to measure localization errors, reflecting both quantization and communication effects. We further use empirical misdetection probability to quantify sensing performance. We combine Wasserstein distance and empirical misdetection probability into a novel cost function, which is inspired by the generalized optimal sub-pattern assignment (GOSPA) metric~\cite{Rahmathullah2017}. We use the proposed GOSPA-like metric to investigate, via numerical simulations, how multi-target localization performance depends on the fraction of blocklength devoted to sensing vs. communication, as well as on the resolution of the quantizer used to describe the positions of the targets.

\subsection{Paper Outline}

The remainder of the paper is organized as follows. In Section~\ref{sec-systmodel}, we introduce the proposed TUMA framework, including the encoder, the channel model, and the type-based decoder. In Section~\ref{sec-proptumadec}, we describe in detail the decoding algorithm, covering the multisource AMP formulation and the modified Bayesian denoiser. We also present an efficient MC approximation of the high-dimensional integrals appearing in the decoding algorithm and a distributed implementation. In Section~\ref{sec:MTL-TUMA}, we discuss how to use TUMA in a multi-target localization scenario, introduce the sensing model and prior, and define the relevant performance metrics. In Section~\ref{Sec:SimResults}, we provide simulation results, unveiling the tradeoff between sensing, quantization, and communication. We conclude the paper in Section~\ref{Sec:Conclusion}.

\subsection{Notation}

We denote system parameters by nonitalic uppercase letters (e.g., $\uA$), sets by calligraphic letters (e.g., $\setS$), scalar variables by italic lowercase letters (e.g., $x$),  vectors by boldface italic lowercase letters (e.g., $\vecx$), and matrices by boldface nonitalic uppercase letters (e.g., $\mathbf{X}$). We write the $(a,b)$th element of $\mathbf{X}$ as $[\mathbf{X}]_{a,b}$, and the $b$th element of $\vecx$ as $[\vecx]_b$. We denote the $n \times n$ identity matrix by $\mathbf{I}_n$, and the all-zero vector by~$\veczero$. Transposition and Hermitian transposition are indicated by the superscripts $\tp{}$ and $\herm{}$, respectively. We denote the complex proper Gaussian vector distribution with mean $\veczero$ and covariance $\mathbf{A}$ by $\mathcal{C}\mathcal{N}(\veczero, \mathbf{A})$, and its probability density function by $\mathcal{C}\mathcal{N}(\cdot; \veczero, \mathbf{A})$. The Binomial distribution with $n$ trials and success probability~$p$ is denoted by $\Bin(n,p)$, and its probability mass function by $\Bin(\cdot; n,p)$. The uniform distribution over the interval $[a,b]$ is denoted by $\Unif(a,b)$ and the $\ell_p$-norm by $\vecnorm{\cdot}_p$; $[n]$ stands for $\set{1,\dots,n}$. We denote the Kronecker delta function by $\delta(\cdot)$, the indicator function by $\mathbbm{1}\{\cdot\}$, the Kronecker product by $\otimes$, and elementwise multiplication by $\odot$; $\diag(x_1, \cdots, x_n)$ is a diagonal matrix with $x_1, \dots, x_n$ as its diagonal entries. The notation $\sim_{\text{i.i.d.}}$ indicates that the elements of a matrix are drawn independently from the specified distribution. Finally, we denote the probability simplex over the set $[\uM]$ by $\mathcal{P}([\uM])$.

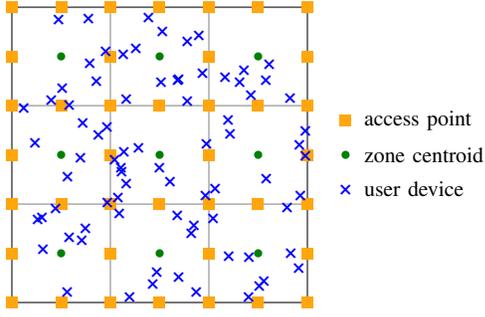
\begin{figure}[t!]
    \centering
    \resizebox{0.75\columnwidth}{!}    {\definecolor{darkgreen}{rgb}{0.0, 0.5, 0.0}
\definecolor{yelloworange}{rgb}{1.0, 0.65, 0.05}
\definecolor{bluecolor}{rgb}{0.0, 0.0, 1.0}
\def\SensorSize{0.25} 

\begin{tikzpicture}

\def\N{4} 
\def\D{4} 
\def\APsize{0.3} 
\def\CrossSize{0.1} 

\draw[black, very thick] (0, 0) rectangle (3*\D, 3*\D);

\foreach \x in {1,2} {
    \draw[gray, thick] (\x*\D, 0) -- (\x*\D, 3*\D);
}
\foreach \y in {1,2} {
    \draw[gray, thick] (0, \y*\D) -- (3*\D, \y*\D);
}

\foreach \x in {0,1,2,3} {
    \foreach \y in {0,1,2,3} {
        \node[fill=yelloworange, rectangle, scale=1.5, minimum width=\APsize cm, minimum height=\APsize cm] at (\x*\D, \y*\D) {};
        \ifnum\y<3
            \node[fill=yelloworange, rectangle, scale=1.5, minimum width=\APsize cm, minimum height=\APsize cm] at (\x*\D, \y*\D + 2) {};
        \fi
        \ifnum\x<3
            \node[fill=yelloworange, rectangle, scale=1.5, minimum width=\APsize cm, minimum height=\APsize cm] at (\x*\D + 2, \y*\D) {};
        \fi
    }
}

\foreach \x in {0,1,2} {
    \foreach \y in {0,1,2} {
        \node[fill=darkgreen, circle, scale=1] at (\x*\D+2, \y*\D+2) {};
    }
}

\foreach \x/\y in {
4.53/10.09, 7.40/3.12, 10.32/9.04, 7.49/0.42, 4.37/3.63, 2.99/3.00, 8.81/1.86, 4.78/0.21, 0.94/6.49, 11.92/6.95, 
4.43/5.46, 8.66/9.16, 7.74/9.33, 1.78/3.82, 4.65/8.28, 7.12/8.19, 3.52/6.80, 3.16/9.72, 6.43/4.89, 7.29/2.80, 
6.78/9.08, 9.43/9.45, 11.68/6.40, 8.78/7.42, 2.84/2.51, 11.92/5.96, 7.91/6.46, 6.12/11.01, 8.26/4.62, 10.05/0.66, 
4.65/4.84, 2.25/0.43, 4.17/5.80, 6.73/9.02, 9.62/0.21, 5.72/0.76, 2.33/8.03, 1.23/3.44, 7.12/10.61, 5.89/1.24, 1.05/3.37, 7.60/10.83, 4.46/5.31, 5.14/6.29, 2.33/2.65, 1.60/8.23, 9.71/8.44, 3.11/11.54, 6.77/1.02, 6.72/3.54, 6.00/5.49, 2.79/5.87, 11.33/1.98, 11.65/1.39, 4.31/4.28, 11.72/4.36, 2.05/8.70, 6.06/8.94, 0.48/7.88, 8.86/6.85, 1.27/2.16, 10.45/9.68, 8.17/3.41, 3.87/7.12, 1.90/11.50, 3.88/4.04, 7.87/4.33, 4.55/6.12, 10.24/0.87, 9.63/1.83, 5.55/11.56, 2.27/5.11, 2.89/7.30, 3.44/8.99, 10.33/5.04, 3.81/10.20, 5.04/10.32, 11.30/8.30, 9.26/8.95, 11.18/3.85
} {
    \node[bluecolor, scale=2.5] at (\x, \y) {\textbf{\texttimes}};
}

\node[fill=yelloworange, rectangle, scale=1.5, minimum width=\APsize cm, minimum height=\APsize cm] at (13.55, 7.4) {};
\node[right, scale=2.5] at (14, 7.4) {access point};
\node[fill=darkgreen, circle, scale=1] at (13.55, 6.0) {};
\node[right, scale=2.5] at (14, 6.0) {zone centroid};
\node[bluecolor, scale=2.5] at (13.55, 4.6) {\textbf{\texttimes}};
\node[right, scale=2.5] at (14, 4.6) {user device};

\end{tikzpicture}}
    \caption{An example of the proposed zone partitioning. In the example, we have $\uU=9$ zones. This partitioning is compatible with the distance-dependent path-loss model described in Section~\ref{SubSec:SimSetup}.}
    \label{fig:topology}
    \vspace{-0.5cm}
\end{figure}

\section{System Model} \label{sec-systmodel}

We consider a D-MIMO network in which $\uB$ access points (APs) are connected to a central processing unit (CPU) via fronthaul links. The APs collaboratively serve single-antenna users, randomly located in a two-dimensional coverage area $\setD$. The area is partitioned into $\uU$ nonoverlapping zones~$\{\setD_u\}_{u\in [\uU]}$, such that $\setD = \bigcup_{u\in[\uU]} \setD_u$ and $\setD_u \cap \setD_{u'} = \emptyset$, $\forall u \neq u'$. 
Each zone contains locations with similar LSFCs, so that users within the same zone experience approximately the same path-loss towards the APs. This partitioning mitigates the effect of path-loss variability~\cite{cakmak_2025_journal}.\footnote{In practice, one can divide $\setD$ into $\uU$ zones by quantizing the LSFC vector at each position in $\setD$ using a codebook of size $\uU$. Each UE then estimates its LSFC coefficients using control signals transmitted by each AP, and selects its zone, and, hence, the corresponding codebook, on the basis of its quantized estimated LSFC vector.} Each AP~$b$ is located at position $\vecnu_b \in \setD$ and equipped with $\uA$ antennas, yielding $\uF = \uA \times \uB$ antennas in total. Fig.~\ref{fig:topology} illustrates an example of the network topology, where the area is divided into a $3\times 3$ grid of square zones, with APs evenly placed along the zone boundaries. While our model and design are applicable to general topologies as long as users with similar path-loss are grouped within the same zone, we use this specific topology in the simulations in Section~\ref{Sec:SimResults}.

We assume that $\uK$ users are present in the system, with $\uK$ fixed and known by the receiver. In each communication round, a subset of the users becomes active. The number of active users, denoted by $K_\mathrm{a}$, is random and unknown to the receiver. Within each zone $u$, the number of active users, denoted by $K_{\mathrm{a},u}$, is also random and unknown, with $K_\mathrm{a} = \sum_{u=1}^\uU K_{\mathrm{a},u}$. Each active user $j$ in zone $u$ selects a message $W_{u,j}$ from the message set $[\uM]$, where $\uM$ is the total number of possible messages. These messages might be quantized representations of application-specific information, such as target positions in multi-target localization or local model updates in federated learning. In Section~\ref{sec:MTL-TUMA}, we will focus on multi-target localization and provide a specific mechanism for user activation and message generation in such a setup. 

The overall TUMA framework, including message encoding, transmission over the fading channel, and decoding at the receiver, is summarized in Fig.~\ref{fig:block-diagram-squeezed}. Next, we provide a detailed description of each block in the figure.

\subsection{Encoder} \label{subsec-encoder}

The system employs a UMA codebook $\mathbf{C} \in \mathbb{C}^{\uN_{\mathrm{c}} \times \adjustedbar{\uM}}$, where $\uN_{\mathrm{c}}$ is the blocklength and $\adjustedbar{\uM} = \uU \cdot \uM$ is the total number of codewords. The codebook is evenly partitioned into zone-specific subcodebooks: $\mathbf{C} = [\mathbf{C}_1, \dots, \mathbf{C}_{\uU}]$, where $\mathbf{C}_u = [\vecc_{u,1}, \dots, \vecc_{u,\uM}] \in \mathbb{C}^{\uN_{\mathrm{c}} \times \uM}$. The set of column vectors of $\mathbf{C}_u$ forms the set of codewords for zone $u$, denoted as $\setC_u = \{\vecc_{u,1},\dots, \vecc_{u,{\uM}}\}$, where $\lVert \vecc_{u,m}\rVert_2^2 = 1$, $\forall u\in[\uU],\, m\in [\uM]$. The encoder $\enc_{u} \sothat [\uM] \rightarrow \setC_u$ maps each message $W_{u,j}$ to the codeword $\enc_{u}(W_{u,j}) = \vecc_{u,W_{u,j}}$. Note that all users within the same zone use the same encoder.

\subsection{Multiplicity and Type} \label{SubSec:MultType}

We denote the number of users transmitting message~$m$ in zone $u$ by $k_{u,m} \in \{0,1,\dots,K_{\mathrm{a},u}\}$. These form the multiplicity vector $\veck_u = \tp{[k_{u,1},\dots,k_{u,{\uM}}]}$ with $ \lVert \veck_u \rVert_1 = K_{\mathrm{a},u}$. The global multiplicity vector is defined as $\veck = \tp{[k_{1},\dots,k_{{\uM}}]}$ with $k_m = \sum_{u=1}^\uU k_{u,m}$. The type is then obtained as the vector of normalized multiplicities, i.e., $\vect = \tp{[t_{1},\dots,t_{{\uM}}]}$ with $t_{m} = k_{m}/K_{\mathrm{a}}$. This vector represents the empirical distribution of transmitted messages across all active users.

\begin{figure}[t!]
    \centering
    \vspace{-0.1cm}
    \begin{tikzpicture}[node distance=0.5cm and 0.7cm, font=\scriptsize]

    
    \node[draw, rectangle, 
    minimum width=0.9cm, minimum height=0.6cm] (enc) {Encoder $\enc_{u}$};
    \node[left=0.3cm of enc, inner sep=0.8pt] (input) {$W_{u,j}$};
    \draw[->] (input) -- (enc.west);

    \node[draw, circle, right=1.35cm of enc, minimum size=0.1cm] (mult) {$\times$};
    \node[above=0.21cm of mult] (np) {\scriptsize $\sqrt{
    \uE_\mathrm{c}}$};
    

    \node[draw, circle, right=0.3cm of mult, minimum size=0.1cm] (channel_mult) {$\times$};
    \node[above=0.2cm of channel_mult] (channel_label) {\scriptsize $\tp{\vech_{u,j}}$};
    
    \node[draw, circle, right=0.3cm of channel_mult, minimum size=0.4cm] (adder) {$\sum$};
    \node[draw, circle, right=0.3cm of adder, minimum size=0.4cm] (adder2) {$\sum$};
    \node[above right=0.1cm of adder2.north east] (w_source_below) {};
    \node[above=0.5cm of w_source_below.center] (w_source) {\scriptsize $\mathbf{W}$};

    \node[below right=0.1cm of adder2.south east] (y_source_above) {};
    \node[below=2.2cm of y_source_above.center] (y_source) {
    };

    \draw (w_source) -- ++(0,-0.25) -- (w_source_below.center);
    \draw[->] (w_source_below.center) -- (adder2.north east);
    \draw (y_source.north) -- (y_source_above.center);
    \draw (y_source_above.center) -- (adder2.south east);
    
    \draw[->] (enc) -- node[above] {\scriptsize $\vecc_{u,W_{u,j}}$} (mult);
    \draw[->] (np) -- (mult); 
    \draw[->] (mult) -- (channel_mult);

    \draw[->] (channel_label) -- (channel_mult);
    \draw[->] (channel_mult) -- (adder);
    \draw[->] (adder) -- (adder2);

    \node[draw, rectangle, left = 2.7cm of y_source.north, minimum width=0.9cm, minimum height=0.6cm] (dec) {Decoder $\dec$}; 
    \node[left= 0.6cm of dec, minimum width=0.4cm, inner sep=1.8pt] (multiplicity) { $\displaystyle 
    \widehat{\vect}    
    $}; 

    \draw[->] (y_source.north) -- node[above] {\scriptsize $\mathbf{Y}$} (dec);
    \draw[->] (dec) -- (multiplicity);

    \node[above right = 0.4cm and -0.85cm of input, 
    red] (user_label) { active user $j$ in zone $u$};
    \node[above right=1.27cm and -0.9cm of input, 
    black] {zone $u$}; 

    \node[above = 1.5cm of adder.center] (upuserstart) {}; 
    \node[below = 0.95cm of adder.center] (downuserstart) {};

    \node[left = 4.5cm of upuserstart, anchor=west] (leftupuserstart) {From active user $1$ in zone $u$};
    \node[below right = -0.075cm and 0.9cm of leftupuserstart.west] (upvdots) {\scalebox{0.8}{$\vdots$}};
    \node[left = 4.5cm of downuserstart, anchor=west] (leftdownuserstart) {From active user $K_{\mathrm{a},u}$ in zone $u$};
    \node[above right = 0.09cm and 0.9cm of leftdownuserstart.west] (downvdots) {\scalebox{0.8}{$\vdots$}};

    \draw (upuserstart.center) -- (leftupuserstart.east);
    \draw (downuserstart.center) -- (leftdownuserstart.east);
    \draw[->] (upuserstart.center) -- (adder);
    \draw[->] (downuserstart.center) -- (adder);

    \node[above = 2.5cm of adder2.center] (upzonestart) {}; 
    \node[below = 1.9cm of adder2.center] (downzonestart) {}; 

    \draw[->] (upzonestart.center) -- (adder2);
    \draw[->] (downzonestart.center) -- (adder2);

    \node[left = 3cm of upzonestart, anchor=west] (leftupzonestart) {From zone $1$};
    \node[below right = -0.01cm and 0.9cm of leftupzonestart.west] (upvdots2) {\scalebox{0.8}{$\vdots$}};
    \node[left = 3cm of downzonestart, anchor=west] (leftdownzonestart) {From zone $\uU$};
    \node[above right = 0.1cm and 0.9cm of leftdownzonestart.west] (downvdots2) {\scalebox{0.8}{$\vdots$}};

    \draw (upzonestart.center) -- (leftupzonestart.east);
    \draw (downzonestart.center) -- (leftdownzonestart.east);

    \begin{scope}[on background layer]
        \node[draw=none, fill=gray!10, rounded corners, fit={
        (input) (enc) (channel_mult) (upvdots) (downvdots) (adder) (leftupuserstart) (leftdownuserstart)}, inner xsep=0.3cm, inner ysep=0.15cm, yshift=-0cm, xshift=-0.19cm] (zoneu) {};
        \node[draw, dashed, red, rounded corners, fit={
        (input) (enc) (user_label) (channel_mult) (channel_label) }, inner xsep=0.15cm, inner ysep=0.1cm, yshift=-0.1cm, xshift=-0.0cm] (zoneu) {};
    \end{scope}
\end{tikzpicture}
    \caption{Block diagram of the proposed TUMA framework with fading channel in a D-MIMO system.}
    \label{fig:block-diagram-squeezed}
\end{figure}
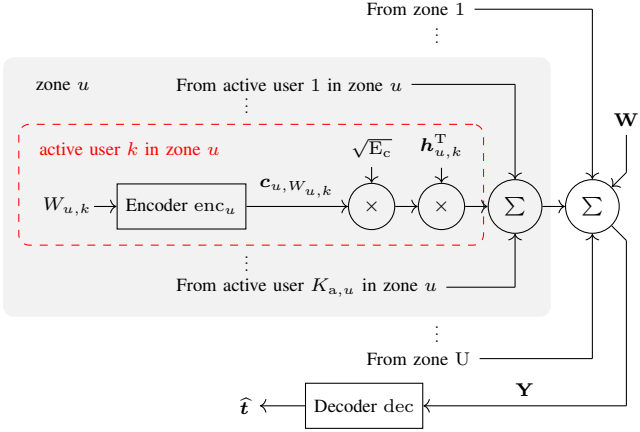

\subsection{Channel Model} \label{subsec:channel_model}

We index by $(u,j)$ the $j$th active user in zone~$u$. The channel between active user~$(u,j)$ and AP~$b$ is modeled as a quasi-static Rayleigh fading channel. Specifically, the channel coefficients are assumed to be independent across antennas, APs, and users, and to stay constant over the duration of a codeword. The channel vector $\vech_{u,j} \in \mathbb{C}^{\uF}$ between active user $(u,j)$ and all receive antennas is distributed as $\mathcal{C}\mathcal{N}(\veczero, \matSigma(\vecrho_{u,j}))$, where $\vecrho_{u,j}$ denotes the position of active user~$(u,j)$ and $\matSigma(\vecrho_{u,j}) = \diag(\gamma_1(\vecrho_{u,j}), \dots, \gamma_{\uB}(\vecrho_{u,j})) \otimes \mathbf{I}_{\uA}$ with $\gamma_b(\vecrho_{u,j})$ being the position-specific LSFC. When multiple users transmit the same codeword, their contributions are superimposed at the receiver. For a codeword~$\vecc_{u,m}$, the \emph{effective} channel vector is given by
\begin{equation}
\vecx_{u,m} = 
    \begin{cases}
    \sum_{j\in[K_{\mathrm{a},u}] \sothat W_{u,j}=m} \vech_{u,j} & \text{if } k_{u,m}>0, \\
    \veczero & \text{if } k_{u,m}=0.
    \end{cases} \label{expression_x} 
\end{equation}
Collecting the effective channel vectors for all codewords in zone $u$, we define the effective channel matrix as $\mathbf{X}_u = \tp{[\vecx_{u,1}, \dots, \vecx_{u,{\uM}}]} \in \mathbb{C}^{{\uM} \times \uF}$. The aggregated received signal across all APs is then given by 
\begin{equation} \label{eqn-mainreceivedsignal}
\mathbf{Y} = \sqrt{\uE_\mathrm{c}}  \sum_{u=1}^{\uU} \mathbf{C}_u \mathbf{X}_u + \mathbf{W} \in \mathbb{C}^{\uN_\mathrm{c} \times \uF},
\end{equation}
where $\mathbf{W} \sim_{\text{i.i.d.}} \mathcal{C}\mathcal{N}(0,\sigma_w^2)$ is the AWGN. The per-codeword average transmit energy is denoted by $\uE_\mathrm{c}$, and the transmit signal to noise ratio (SNR) is defined as $\text{SNR}_{\text{tx}} = \uE_\mathrm{c}/ (\uN_\mathrm{c}\sigma_w^2)$.

\subsection{Decoder} \label{Sec:SystemModel-Decoding}

The decoder estimates the message type by using a decoding function defined as $\dec \sothat \mathbb{C}^{\uN_{\mathrm{c}} \times \uF} \rightarrow \setP([{\uM}])$. Given $\mathbf{Y}$ and $\mathbf{C}$, the decoder first estimates the per-zone multiplicities, $\widehat{\veck}_u = \tp{[\widehat{k}_{u,1}, \dots, \widehat{k}_{u,{\uM}}]}$, $u=1,\dots,\uU$, and then computes the global multiplicity vector, $\widehat{\veck} = \sum_{u=1}^\uU \widehat{\veck}_u$.  
Finally, the type is estimated as $\widehat{\vect} = \widehat{\veck} / \lVert \widehat{\veck} \rVert_1$. The performance of type estimation is evaluated using the average TV distance between the type of the transmitted messages and its estimate, defined as
\begin{equation}
\dTV = \frac{1}{2} \Exop \left[  \sum_{m=1}^{{\uM}} | t_{m} - \widehat{t}_{m}| \right], \label{eq:tv_dist_expression}
\end{equation}
where the expectation is over the randomness of messages, user positions, small-scale fading, and additive noise.

\section{Proposed TUMA Decoder} \label{sec-proptumadec}

For the TUMA framework just introduced, we propose a decoder that employs the multisource AMP algorithm developed in~\cite{cakmak_2025_journal}. 
AMP is an iterative algorithm for estimating unknown vectors in large linear systems. In each AMP iteration, a linear residual update is performed, followed by a nonlinear Bayesian denoising operation, together with an Onsager correction that mitigates the correlations created across iterations. This ensures that the effective noise remains Gaussian, which enables accurate asymptotic performance prediction through state evolution~\cite{Bayati_2011}. 
In the multisource AMP framework, the term \emph{multisource} refers to the partitioning of the global codebook into multiple subcodebooks, each treated as a separate signal source. Users are assigned to these sources by selecting the subcodebook corresponding to their LSFCs. The multisource AMP algorithm then processes these sources in parallel within the same iterative procedure. This differs from classical AMP, which operates on a single unpartitioned codebook. The analysis in~\cite{cakmak_2025_journal} establishes that, under Gaussian codebooks, the behavior of multisource AMP in a D-MIMO architecture is accurately predicted by state evolution in the limit  $\uN_\mathrm{c} \rightarrow \infty$ with $\uM/\uN_\mathrm{c}$ fixed.

In the TUMA setting, message collisions introduce additional uncertainty: the number of users contributing to each codeword, and their positions within the zone, are both unknown. We account for this by modifying the Bayesian prior so that each unknown effective channel vector $\vecx_{u,m}$ is modeled as a weighted mixture over all possible multiplicities and user positions within zone~$u$. Based on this prior, we derive a multiplicity-aware denoiser that computes the posterior mean of the effective channel $\vecx_{u,m}$ for each message~$m$. We also modify the Onsager correction term to ensure consistency with the new denoiser. After the AMP iterations, we estimate the message multiplicities by using their posterior probabilities, computed in the denoising step. The details are provided below. 

Throughout, we assume that the receiver has access to the received signal~$\mathbf{Y}$, the codebook~$\mathbf{C}$, the transmit energy~$\uE_\mathrm{c}$, the LSFC model~$\gamma_b(\cdot)$, and the total number of users~$\uK$. However, the decoder treats the number of active users in each zone~$\{K_{\mathrm{a},u}\}_{u\in[\uU]}$, as well as their positions, as unknown random variables. Furthermore, no instantaneous CSI is available anywhere in the network. In the following, we introduce a centralized decoder and discuss approximations for its efficient implementation.
We then propose a distributed version of this decoder, which is inspired by dAMP~\cite{bai_larsson_damp}, which was originally designed for activity detection with perfect per-user LSFC information. Its adaptation to our setting, which involves unknown user positions and message multiplicities, requires nontrivial modifications.

\subsection{Centralized Decoder} \label{SubSec-CentDec}
The centralized decoder employs the multisource AMP algorithm~\cite{cakmak_2025_journal} to iteratively process the received signal and extract necessary information for multiplicity estimation.

\subsubsection{Multisource AMP} \label{subsec:multi_amp}
We recap the multisource AMP algorithm from~\cite{cakmak_2025_journal} and refer the reader to that work for additional details including its high-dimensional analysis. The algorithm performs $\uT_\text{AMP}$ iterations, where at iteration $t$, $\mathbf{X}_u^{(t)}$ denotes the estimate of the effective channel matrix for zone $u$, and $\mathbf{Z}^{(t)}$ denotes the residual noise. These matrices are initialized as $\mathbf{X}_u^{(0)} = \veczero$ and $\mathbf{Z}^{(0)}=\mathbf{Y}$, and updated as follows
\begin{subequations} \label{amp_decoder}
    \begin{align}
        &\mathbf{R}_u^{(t)} = \mathbf{C}_u^{\uH} \mathbf{Z}^{(t-1)} + \sqrt{\uE_\mathrm{c}}
        \mathbf{X}_u^{(t-1)} \in \mathbb{C}^{\uM \times \uF}, \\
        &\mathbf{X}_u^{(t)} = \eta_{u,t}(\mathbf{R}_u^{(t)}) \in \mathbb{C}^{\uM \times \uF}, \label{amp_decoder_row2} \\
        &\mathbf{\Gamma}_u^{(t)} = \mathbf{C}_u \mathbf{X}_u^{(t)} - \frac{{\uM}}{\uN_{\mathrm{c}}} \mathbf{Z}^{(t-1)} \mathbf{Q}_u^{(t)} \in \mathbb{C}^{\uN_\mathrm{c} \times \uF}, \label{amp_decoder_row3} \\
        &\mathbf{Z}^{(t)} = \mathbf{Y} - \sqrt{\uE_\mathrm{c}} 
        \sum_{u=1}^{\uU} \mathbf{\Gamma}_u^{(t)} \in \mathbb{C}^{\uN_\mathrm{c} \times \uF}.
    \end{align}
\end{subequations}
Here, $\mathbf{Q}_u^{(t)}$ is the so-called Onsager term~\cite{cakmak_2025_journal}, which will be described in Section~\ref{Subsec:onsager}. The denoiser $\eta_{u,t}(\cdot)$ acts row-wise on $\mathbf{R}_u^{(t)}$. The following theorem establishes that the multisource AMP
state-evolution characterization in~\cite{cakmak_2025_journal}
remains valid under the message-collision model considered in TUMA.

\begin{theorem}[State evolution of multisource AMP under message collision]
\label{thm:tuma_se}
Consider the TUMA model in~\eqref{eqn-mainreceivedsignal}  and the AMP recursion in~\eqref{amp_decoder}. Assume that
$\mathbf C_u\in\mathbb C^{\uN_{\mathrm c}\times\uM}$ has i.i.d.
$\mathcal{CN}(0,1/\uN_{\mathrm c})$ entries,\footnote{This assumption does not satisfy the deterministic power constraint $\lVert \vecc_{u,m} \rVert_2^2=1$ imposed in Section~\ref{subsec:channel_model}. 
However, in our simulations, we have not observed performance differences between Gaussian and normalized Gaussian codebooks under the simulation settings of Section~\ref{SubSec:SimSetup}. Therefore, in the numerical results, we employ a normalized Gaussian codebook, obtained by normalizing each Gaussian codeword to unit norm.} that
$\alpha=\uM/\uN_{\mathrm c}$, $\uU$, $\uF$, and $\uE_{\mathrm c}$ remain fixed as
$\uN_{\mathrm c},\uM\to\infty$, and that the denoisers
$\eta_{u,t}(\cdot)$ are differentiable and Lipschitz continuous for all $u$ and $t$.

Let $\vecx_u$ denote a generic row of $\mathbf X_u$. 
Assume further that the message multiplicities satisfy
$k_{u,m}\in\{0,\ldots,\uK_{\text{max}}\}$ for all $(u,m)$
and that, conditioned on the multiplicities and user positions, the
channel vectors appearing in~\eqref{expression_x} are independent
zero-mean Gaussian vectors with uniformly bounded covariance matrices.

Then, for every fixed AMP iteration $t$, each row of $\mathbf R_u^{(t)}$ converges, in the high-dimensional sense as described in~\cite[Thm.~1]{cakmak_2025_journal}, to
\begin{equation}
\vecr_u^{(t)}
=
\sqrt{\uE_c}\,\vecx_u
+
\vecvarphi_u^{(t)}, \label{decouple_model}
\end{equation}
where the effective noise 
$\vecvarphi_u^{(t)}
\sim
\mathcal{CN}(\veczero,\mathbf T^{(t)})$
is independent of~$\vecx_u$. Moreover, its covariance matrix is given by $\mathbf T^{(t)}=\mathbf T^{(t,t)}$, where 
\begin{equation}
\mathbf{T}^{(t,s)}
=
\sigma_w^2 \mathbf I_{\uF}
+ 
\sum_{u=1}^{\uU} \mathbf{T}_u^{(t,s)},
\label{eq:T_ts}
\end{equation}
with
\begin{equation}
\mathbf{T}_u^{(t,s)}
=
\alpha \uE_{\mathrm{c}} \,
\mathbb{E} \mathopen{} \left[
(\vecx_u - \vecx_u^{(t)})
(\vecx_u - \vecx_u^{(s)})^{\uH}
\right]
\in \mathbb{C}^{\uF \times \uF}.
\label{eq:Tu_ts}
\end{equation}
Here, \(\vecx_u^{(t)} \triangleq \eta_{u,t}(\sqrt{\uE_\mathrm{c}} \, \vecx_u + \vecvarphi_u^{(t)}) \).

\end{theorem}

\begin{proof}

The result follows by specializing the multisource AMP high-dimensional analysis in~\cite[Def.~1, Assumption~1, Thm.~1]{cakmak_2025_journal} to the distribution of a generic row of $\mathbf{X}_u$ induced by the TUMA multiplicity model.

We start by noting that the AMP recursion in~\eqref{amp_decoder} has the same multisource AMP structure as the recursion analyzed in~\cite{cakmak_2025_journal}, up to the deterministic signal scaling by $\sqrt{\uE_{\mathrm c}}$. This scaling is absorbed into the the distribution of $\vecx_u$ and the state-evolution covariance recursion through the factor $\uE_{\mathrm c}$ in~\eqref{eq:Tu_ts}.

Second, the distribution of $\vecx_u$ satisfies the bounded-moment condition required in~\cite[Assumption~1]{cakmak_2025_journal}. Indeed, it follows from~\eqref{expression_x} that conditioned on a realization of the
multiplicity~$k_{u,m}$ and the corresponding user positions,
a generic row $\vecx_u$ is a finite sum of independent Gaussian
channel vectors and is therefore Gaussian.
Since $k_{u,m}\le \uK_{\text{max}}$, the row distribution is a finite
Gaussian mixture with bounded moments of all orders.

Finally, the denoisers are assumed to be differentiable and Lipschitz continuous, as required by~\cite[Thm.~1]{cakmak_2025_journal}. Therefore, all assumptions of the multisource AMP state-evolution theorem~\cite[Thm.~1]{cakmak_2025_journal} are satisfied. This implies that the Gaussian decoupling in~\eqref{decouple_model} and the covariance recursion in~\eqref{eq:Tu_ts}--\eqref{eq:T_ts} hold also in our TUMA setup.
\end{proof}

Theorem~\ref{thm:tuma_se} implies that each row of $\mathbf R_u^{(t)}$ behaves asymptotically as an observation from the Gaussian model~\eqref{decouple_model}  with effective noise covariance $\mathbf T^{(t)}=\mathbf T^{(t,t)}$. Accordingly, we pick the denoiser $\eta_{u,t}(\cdot)$ as the Bayesian posterior mean estimator (PME) associated with the decoupled model in~\eqref{decouple_model}.

Combining the state evolution equations~\eqref{eq:Tu_ts}--\eqref{eq:T_ts} with the high-dimensional representation of the AMP residual in\cite[Eq.~(18)]{cakmak_2025_journal}, one can show that, in the large-system limit, using the high-dimensional equivalence introduced in~\cite{cakmak_2025_journal} together with the Borel-Cantelli lemma~\cite[Sec.~4]{borelcantelli_book}, 
\begin{equation}
\frac{1}{\uN_{\mathrm c}}
(\mathbf Z^{(t-1)})^{\uH}
\mathbf Z^{(t-1)}
\xrightarrow[]{\mathrm{a.s.}}
\mathbf T^{(t)},
\label{eq:residual_cov_converges}
\end{equation}
i.e., the empirical covariance of the residual $\mathbf{Z}^{(t-1)}$ converges almost surely to the effective noise covariance $\mathbf{T}^{(t)}$. 
Because the block diagonal structure in the covariance matrix is preserved by the state evolution recursion~\cite[Sec.~III]{cakmak_2025_journal}, we have that $\mathbf{T}^{(t)}=\diag\!\big(\tau_1^{(t)},\dots,\tau_{\uB}^{(t)}\big)
\otimes \mathbf{I}_{\uA}$, where $\tau_b^{(t)}$ denotes the effective noise variance associated with AP~$b$ at iteration~$t$. 
In practice, we estimate these variances from the residual as
\begin{equation}
    \tau_b^{(t)} =
    \frac{1}{\uA \uN_\mathrm{c}}
    \sum_{a=1}^{\uA}
    \Re\!\left\{
    \big[
      (\mathbf{Z}^{(t-1)})^{\uH}\mathbf{Z}^{(t-1)}
    \big]_{(b-1)\uA + a,\,(b-1)\uA + a}
    \right\},
    \label{residual_expression}
\end{equation}
which is a sample-average estimate of the diagonal entries of $\mathbf{T}^{(t)}$ over the antennas of AP~$b$. The real-part operator $\Re\{\cdot\}$ is included as a practical implementation detail to compensate for  numerical errors.

\subsubsection{Prior Selection} \label{subsec:prior} To derive the PME for the model in~\eqref{decouple_model}, we need a prior on $\vecx_{u,m}$. This in turn requires one to specify a prior on the message multiplicities, i.e., a probability distribution $p(k_{u,m}=k)$, for $k=0,\dots,\uK$. This prior is application-dependent, as it is influenced by factors such as the underlying user activity and message selection mechanisms. We provide one such prior for a multi-target localization application in Section~\ref{sec:MTL-TUMA}.

We assume that the active users in each zone~$u$ are independently and uniformly distributed over the zone. Given $k_{u,m}=k$, the positions of users transmitting the $m$th message in zone $u$ are denoted by $\vecrho_{u,m,1:k} = \tp{[\tp{\vecrho_{u,m,1}},\ldots,\tp{\vecrho_{u,m,k}}]}$, where each $\vecrho_{u,m,i}$ is a point in $\setD_u$.\footnote{We treat all positions as 2D coordinates for simplicity. An extension to 3D coordinates is straightforward.} These positions follow the distribution $p(\vecrho_{u,m,1:k_{u,m}} \mid k_{u,m}=k) = 1/|\setD_u|^k$, where $|\setD_u|$ denotes the area of the region $\setD_u$.

\subsubsection{Denoiser} \label{subsec:denoiser} The Bayesian PME of $\vecx_{u,m}$ given $\mathbf{R}_u^{(t)}$ is derived using the large-system approximation~\eqref{decouple_model}. For simplicity, the iteration index~$(t)$ is omitted in the following equations. We can express the PME denoiser as
\begin{equation}
    \eta_u (\vecr_{u,m}) = \sum_{k=1}^{\uK} \Exop[\vecx_{u,m} | \vecr_{u,m}, k_{u,m}=k] \, p(k_{u,m}=k \mid \vecr_{u,m}). \label{eq:start_of_posterior_denoiser_b}
\end{equation}
We next compute both terms on the right-hand side (RHS) of~\eqref{eq:start_of_posterior_denoiser_b}. For notational simplicity, we omit the subscript $m$ and use $k_u$, $\vecrho_{u,1:k_u}$, $\vecr_u$, and $\vecx_u$ to denote $k_{u,m}$, $\vecrho_{u,m,1:k_{u,m}}$, $\vecr_{u,m}$, and $\vecx_{u,m}$, respectively. Using Bayes' theorem, we can express the second factor on the RHS of~\eqref{eq:start_of_posterior_denoiser_b}~as
\begin{equation} \label{eqn:posterior_mults}
    p(k_{u}=k \mid \vecr_{u}) = \frac{p(\vecr_{u} \mid k_{u}=k) \, p(k_{u}=k)}{\sum_{\ell=0}^{\uK} p(\vecr_{u} \mid k_{u}=\ell) \, p(k_{u}=\ell)}.
\end{equation}
The likelihood $p(\vecr_{u} \mid k_{u}=k)$ is given by 
\begin{subequations}
    \begin{align}
    &p(\vecr_u \mid k_u=k) \notag \\
    &=\int_{\setD_u^{k}} p(\vecr_u \mid \vecrho_{u,1:k}, k_u=k) \, p(\vecrho_{u,1:k} \mid k_u=k)\, \mathrm{d}\vecrho_{u,1:k} \label{eqn-deriv-p_ruku_step1}\\
    &=\frac{1}{|\setD_u|^{k}} \int_{\setD_u^{k}} p(\vecr_u \mid \vecrho_{u,1:k})\, \mathrm{d}\vecrho_{u,1:k}, \label{eqn:likelihood_r_given_k}
    \end{align}
\end{subequations}
where $\setD_u^{k} = \setD_u \times \cdots \times \setD_u $. In \eqref{eqn-deriv-p_ruku_step1}, we used the law of total probability, and in~\eqref{eqn:likelihood_r_given_k}, we used the Markov property $k_{u} \leftrightarrow \vecrho_{u,1:k_{u}} \leftrightarrow \vecx_{u} \leftrightarrow \vecr_{u}$ and the assumption that $p(\vecrho_{u,1:k_u} \mid k_u=k)= 1/|\setD_u|^{k}$. It follows from (\ref{expression_x}) and \eqref{decouple_model} that, in the large-system limit and for the case $k_u=k$, the likelihood $p(\vecr_{u} \mid \vecrho_{u,1:k})$ is given by
\begin{equation}   \label{eqn:explicit_likelihood_expression}
p(\vecr_{u} \mid \vecrho_{u,1:k}) = \textstyle \mathcal{C}\mathcal{N}(\vecr_{u}; \veczero, \mathbf{T} + \uE_\mathrm{c} 
\sum_{i=1}^{k} \matSigma(\vecrho_{u,i})).
\end{equation}
Next, the conditional mean $\Exop[\vecx_{u} \mid \vecr_{u}, k_{u}=k]$ in \eqref{eq:start_of_posterior_denoiser_b} can be expressed as
\begin{align} 
    \Exop[\vecx_{u} &\mid \vecr_{u}, k_{u}=k] \notag \\
    & = \int_{\setD_u^{k}} \Exop[\vecx_{u} \mid \vecr_{u}, \vecrho_{u,1:k}] \, p(\vecrho_{u,1:k} \mid \vecr_{u})\, \mathrm{d}\vecrho_{u,1:k}. \label{eqn-bayesian-pme-17}
\end{align}
Note now that \eqref{decouple_model}, together with the Gaussianity of $\vecx_u$, imply that~\cite[Sec.~12.5]{mmse_est} 
\begin{align} \label{eqn:mmse_est}
    &\Exop[\vecx_{u} \mid \vecr_{u}, \vecrho_{u,1:k}] \notag \\
    &= \textstyle  ( 
    \sqrt{\uE_\mathrm{c}} 
    \sum_{i=1}^{k} \matSigma(\vecrho_{u,i}) ) (\mathbf{T} + 
     \uE_\mathrm{c}
    \sum_{i=1}^{k} \matSigma(\vecrho_{u,i}))^{-1} \vecr_{u} .
\end{align}
Finally, we evaluate the posterior of the user positions in \eqref{eqn-bayesian-pme-17}  as
\begin{subequations} \label{eqn-bayesian-pme-19-main}
\begin{align}
p&(\vecrho_{u, 1:k} \mid \vecr_u) =p(\vecrho_{u,1:k_u} \mid \vecr_u, k_u=k) \notag \\ 
&= \frac{p(\vecr_u\mid \vecrho_{u, 1:k_u},k_u=k)p(\vecrho_{u,1:k_u}\mid k_u=k)}{p(\vecr_u\mid k_u=k)} \label{eqn-deriv-posteriorpositions_step1} \\
&= \frac{p(\vecr_u\mid \vecrho_{u,1:k})p(\vecrho_{u,1:k})}{\int_{\setD_u^{k}}p(\vecr_u\mid \vecrho'_{u,1:k}) p(\vecrho'_{u,1:k}\mid k_u=k) \, \mathrm{d}\vecrho'_{u,1:k}} \label{eqn-deriv-posteriorpositions_step2} \\
&= \frac{p(\vecr_u\mid \vecrho_{u,1:k})/|\setD_u|^{k}}{\int_{\setD_u^{k}}(p(\vecr_u\mid \vecrho'_{u,1:k}) /|\setD_u|^{k}) \, \mathrm{d}\vecrho'_{u,1:k}} \label{eqn-deriv-posteriorpositions_step3} \\
&= \frac{p(\vecr_u\mid \vecrho_{u,1:k})}{\int_{\setD_u^{k}}p(\vecr_u\mid \vecrho'_{u,1:k}) \, \mathrm{d}\vecrho'_{u,1:k}}. \label{eqn-bayesian-pme-19}
\end{align}
\end{subequations}
Here, \eqref{eqn-deriv-posteriorpositions_step1} follows from the Bayes' theorem, \eqref{eqn-deriv-posteriorpositions_step2} from the law of total probability and the Markov property $k_{u} \leftrightarrow \vecrho_{u,1:k_{u}} \leftrightarrow \vecx_{u} \leftrightarrow \vecr_{u}$, and \eqref{eqn-deriv-posteriorpositions_step3} from the assumption of a uniform prior.

\subsubsection{Onsager Correction} \label{Subsec:onsager}
In multisource AMP~\cite{cakmak_2025_journal}, the Onsager term $\mathbf{Q}_u^{(t)} \in \mathbb{C}^{\uF \times \uF}$ in (\ref{amp_decoder_row3}) is given by
\begin{align}
[\mathbf{Q}_u^{(t)}]_{a,b} = \frac{1}{{\uM}} \sum_{m=1}^{{\uM}} \frac{\partial [\eta_{u,t}(\vecr_{u,m}^{(t)})]_b}{\partial [\vecr_{u,m}^{(t)}]_a}. \label{onsager_expression}
\end{align}
A numerically computable expression for this term is provided in Appendix~\ref{app:onsager}.

\subsubsection{Multiplicity Estimation} \label{subsec:type_est} We perform maximum a posteriori decoding as
\begin{align} \label{type_estimation_mapml} 
\widehat{k}_{u,m}  
&= \argmax_{k \in \{0,1,\ldots,\uK\}} p(k_{u,m}=k\mid \vecr_{u,m}). 
\end{align} 
Then, we estimate the type $\widehat{\vect}$ as outlined in Section~\ref{Sec:SystemModel-Decoding}.

\subsection{Approximations for Efficient Implementation}
The PME denoiser involves multi-dimensional integrals in \eqref{eqn:likelihood_r_given_k}, \eqref{eqn-bayesian-pme-17}, and \eqref{eqn-bayesian-pme-19} whose evaluation becomes computationally prohibitive for large values of multiplicities. One could discretize the coverage area using a uniform discrete grid. However, the complexity of this approach still grows exponentially with the multiplicity. Instead, we adopt MC sampling. Specifically, we draw user positions independently from the uniform prior over $\setD_u$. Let $\{\vecrho^i_{u,1:k}\}_{i=1}^{\uN_{\text{MC}}}$ denote the MC samples of the positions of $k$ users. Using these samples, we approximate \eqref{eqn:posterior_mults} as 
\begin{equation}
    p(k_u=k \mid \vecr_u) \approx \frac{\sum_{i=1}^{\uN_{\text{MC}}} p(\vecr_u \mid \vecrho^i_{u,1:k}) p(k_u=k)}{\sum_{i=1}^{\uN_{\text{MC}}} \sum_{\ell=0}^{\uK} p(\vecr_u \mid \vecrho^i_{u,1:\ell}) p(k_u=\ell)}. \label{eqn-MCsampling2}
\end{equation}
and \eqref{eqn-bayesian-pme-17} as 
\begin{equation}
    \Exop[\vecx_u \mid \vecr_u, k_u=k] \approx \frac{\sum_{i=1}^{\uN_{\text{MC}}} \Exop[\vecx_u \mid \vecr_u, \vecrho^i_{u,1:k}] p(\vecr_u \mid \vecrho^i_{u,1:k})}{\sum_{i=1}^{\uN_{\text{MC}}} p(\vecr_u \mid \vecrho^i_{u,1:k})}. \label{eqn-MCsampling1}
\end{equation}
To further reduce complexity, we limit the computation to a maximum multiplicity $\uK_\text{max}$, i.e., we replace $\uK$ in equations \eqref{eq:start_of_posterior_denoiser_b}, \eqref{eqn:posterior_mults}, \eqref{type_estimation_mapml}, and \eqref{eqn-MCsampling1} by $\uK_\text{max}$. The choice of $\uK_\text{max}$ can be guided by the prior $p(k_{u,m}=k)$, to ensure that only statistically significant multiplicities are considered. We summarize the proposed centralized decoder with MC sampling-based approximation in Algorithm~\ref{CentralizedDecoder}.

\begin{algorithm}[t!]
\small
\caption{Centralized Decoder with Monte Carlo Sampling} 
\textbf{Hyperparameters:} Number of MC samples~$\uN_\mathrm{MC}$, 
number of AMP iterations~$\uT_\text{AMP}$, maximum multiplicity~$\uK_{\text{max}}$ \\
\textbf{Inputs:} Received signal $\mathbf{Y}$, codebook $\mathbf{C}$, transmit energy~$\uE_\mathrm{c}$, sampled positions $\{\vecrho^i_{u,1:k}\}_{i=1}^{\uN_{\text{MC}}}$ for all $k \in [\uK_\text{max}]$, $u \in [\uU]$ \\
\textbf{Output:} Estimated type vector $\widehat{\vect}$ \\
\textbf{Initialization:} $\mathbf{Z}^{(0)} = \mathbf{Y}$, $\mathbf{X}_u^{(0)} = \veczero$, for all $u \in [\uU]$
\begin{algorithmic}[1]
    \vspace{0.1cm}
    \Statex \textbf{1. AMP for Channel Estimation:}
    \State Precompute  $\{\sum_{j=1}^{k}\matSigma(\vecrho^i_{u,j})\}_{i=1}^{\uN_{\text{MC}}}$ for $k \in [\uK_\text{max}]$, $u \in [\uU]$
    \For{$t \gets 1$ to $\uT_\text{AMP}$}
        \For{$u \gets 1$ to $\uU$}  
            \vspace{0.1cm}
            \Statex \hspace{0.85cm} \textit{Effective observation:}
            \State $\mathbf{R}_u^{(t)} \gets \mathbf{C}_u^{\uH} \mathbf{Z}^{(t-1)} + \sqrt{\uE_\mathrm{c}} \mathbf{X}_u^{(t-1)}$
            \vspace{0.1cm}
            \Statex \hspace{0.85cm} \textit{Bayesian denoising:}
            \State Compute covariance matrix $\mathbf{T}^{(t)}$ as in \eqref{residual_expression} 
            \State $\mathbf{X}_u^{(t)} \gets \eta_{u,t}(\mathbf{R}_u^{(t)})$ using \eqref{eq:start_of_posterior_denoiser_b}--\eqref{eqn-bayesian-pme-19-main}, \eqref{eqn-MCsampling2} and \eqref{eqn-MCsampling1} 
            \vspace{0.1cm}
            \Statex \hspace{0.85cm} \textit{Onsager correction:}
            \State Compute $\mathbf{Q}_u^{(t)}$ as in (\ref{onsager_expression}) and Appendix~\ref{app:onsager}
            \State $\mathbf{\Gamma}_u^{(t)} \gets \mathbf{C}_u \mathbf{X}_u^{(t)} - \frac{{\uM}}{\uN_\mathrm{c}} \mathbf{Z}^{(t-1)} \mathbf{Q}_u^{(t)}$
        \EndFor
        \vspace{0.1cm}
        \Statex \hspace{0.4cm} \textit{Residual update:}
        \State $\mathbf{Z}^{(t)} \gets \mathbf{Y} - \sqrt{\uE_\mathrm{c}} \sum_{u=1}^{\uU} \mathbf{\Gamma}_u^{(t)}$
    \EndFor
    \vspace{0.1cm}
    \Statex \textbf{2. Type Estimation:}
    \State Estimate message multiplicities $\{\widehat{\veck}_{u}\}_{u\in [\uU]}$ as in (\ref{type_estimation_mapml}) using $p(k_{u,m}=k \mid \vecr_{u,m})$ computed in line 6 with \eqref{eqn-MCsampling2}
    \State $\widehat{\vect} \gets \sum_{u=1}^\uU \widehat{\veck}_u / \lVert \sum_{u=1}^\uU \widehat{\veck}_u\rVert_1$
\end{algorithmic}
\label{CentralizedDecoder}
\end{algorithm}

\subsection{Distributed Decoder}
To improve scalability in D-MIMO systems~\cite{cf_mimo_book}, we propose a distributed decoder inspired by dAMP~\cite{bai_larsson_damp}. Rather than assuming that the received signal is entirely processed at the CPU, we consider a scenario in which each AP locally processes its received signal and transmits the likelihoods of each codeword to the CPU. Then, the likelihoods are aggregated as (see Appendix~\ref{app:dist_AMP})
\begin{equation} p(\vecr_{u,m} \mid \vecrho_{u,1:k}) = \prod_{b=1}^{\uB} p_b(\vecr_{b,u,m} \mid \vecrho_{u,1:k}), \end{equation} where $p_b(\vecr_{b,u,m} \mid \vecrho_{u,1:k})$ is the local likelihood computed at AP $b$ with MC sampling as in (\ref{eqn-MCsampling2}). The posterior probability is then computed as 
\begin{align} 
&p(k_{u,m}=k \mid \vecr_{u,m}) \notag\\ & \qquad  = \frac{p(k_{u,m}=k) \prod_{b=1}^{\uB} p_b(\vecr_{b,u,m} \mid \vecrho_{u,1:k})}{\sum_{\ell=0}^{\uK} p(k_{u,m}=\ell) \prod_{b=1}^{\uB} p_b(\vecr_{b,u,m} \mid \vecrho_{u,1:\ell})}. \label{eq_dist_likelihood}
\end{align}
as detailed in Appendix~\ref{app:dist_AMP}. 
This design improves scalability by reducing the CPU's workload through local processing at APs. As before, we can replace $\uK$ with $\uK_\text{max}$ in \eqref{eq_dist_likelihood}. The CPU then performs posterior computation and type estimation following the steps already detailed in Section~\ref{SubSec-CentDec}.

\subsection{Complexity Analysis}
The per-iteration complexity of the centralized decoder is dominated by two operations: (i) the linear AMP updates involve multiplications by $\mathbf{C}_u$ and $\mathbf{C}_u^{\mathrm{H}}$ across all zones, which scale as $O(\adjustedbar{\uM}\cdot \uN_\mathrm{c})$; (ii) the PME denoising operation, which involves, for each of the $\adjustedbar{\uM}$ codewords, $\uN_{\text{MC}}$ MC samples for each of the at most $\uK$ multiplicity hypotheses and $\uF$ receive antennas. Because the posterior covariances are diagonal, each likelihood evaluation is $O(\uF)$, leading to $O(\adjustedbar{\uM}\cdot\uN_{\text{MC}}\cdot\uK\cdot\uF)$ operations.
Combining these terms, we obtain a complexity of order $O\big(\adjustedbar{\uM} \cdot (\uN_\mathrm{c} + \uN_{\text{MC}} \cdot \uK \cdot \uF)\big)$.
By limiting the analysis to a maximum multiplicity $\uK_{\text{max}}$, we can reduce the complexity to $O\big(\adjustedbar{\uM} \cdot (\uN_\mathrm{c} + \uN_{\text{MC}} \cdot \uK_{\text{max}} \cdot \uF)\big)$. 

In the distributed decoder, each AP operates on a subset of the data, reducing the per-device computational load. Specifically, the per-AP complexity is of order $O\big(\adjustedbar{\uM} \cdot (\uN_\mathrm{c} + \uN_{\text{MC}} \cdot \uK_{\text{max}} \cdot \uA)\big)$. After local processing, the CPU aggregates the likelihoods and performs type estimation with complexity $O\big(\adjustedbar{\uM} \cdot \uB\big)$. Since $\uA \ll \uF$ in typical D-MIMO systems, the distributed approach reduces significantly the computational cost.

\section{Multi-Target Localization with TUMA} \label{sec:MTL-TUMA}

We apply the TUMA framework proposed in Sections~\ref{sec-systmodel} and~\ref{sec-proptumadec} to a multi-target localization scenario with $\uK$ sensors and $\uT$ targets in the coverage area~$\setD$. The sensors, which act as users, aim to detect the targets and report their positions to the APs using TUMA. The target positions are denoted as $\setT = \{\vecp_1, \dots, \vecp_\uT\}\in\setD^\uT$, and the sensor positions as $\setS = \{\vecs_1,\dots,\vecs_\uK\} \in \setD^\uK$. Each sensor detects the targets using a probabilistic sensing model and quantizes the positions of the detected targets before transmission. Throughout, we assume that each sensor may report at most a single target position: if more than one target is detected, the sensor reports only the closest detected target. We now detail each step of the multi-target localization framework.

\subsection{Sensor Activation and Message Generation}

We define the sensing function for sensor $j$ in zone $u$ as $\sens_{u,j} : \setD^\uT \rightarrow \setD \cup \{\emptyset\}$. It operates as follows:
\begin{itemize}
    \item $\sens_{u,j}(\setT) = \emptyset$ indicates that the sensor $(u,j)$ does not detect any target and therefore remains inactive.    
    \item $\sens_{u,j}(\setT) = \vecp_{O_{u,j}}$ means that the sensor $(u,j)$ detects a target at position $\vecp_{O_{u,j}} \in \setD$, where $O_{u,j} \in [\uT]$ denotes the index of the detected target.
\end{itemize}

We denote by $K_{\mathrm{a},u}$ the number of active sensors in zone u, i.e., the number of sensors that detect a target. Without loss of generality, we assume that active sensors in zone $u$ occupy the first $K_{\mathrm{a},u}$ indices. That is, we reorder sensor indices so that $\sens_{u,j}(\setT) \neq \emptyset$, $\forall j \leq K_{\mathrm{a},u}$ and $\sens_{u,j}(\setT) = \emptyset$, $ \forall j > K_{\mathrm{a},u}$. This simplifies notation and allows us to directly index the active sensors in zone $u$ as $(u,j)$ with $j\in [K_{\mathrm{a},u}]$.

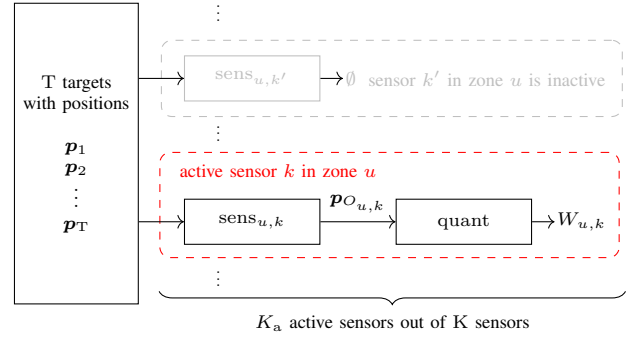
\begin{figure}[!t]
\centering
\begin{tikzpicture}[node distance=0.5cm and 0.7cm, font=\scriptsize]

    \node[rectangle, draw, black, minimum height = 4cm, minimum width = 1cm, node distance=3.7cm, align=center] (targets) at (0,0) {$\uT$ targets \\ with positions \\ \\ $\vecp_1$\\$\vecp_2$ \\ $ \vdots$ \\ $\vecp_{{\uT}}$};

    \node[rectangle, draw, black, minimum height = 0.6cm, minimum width = 1.8cm, node distance=3.7cm, align=center, below right = 0.6cm and 0.6cm of targets.east] (the_sensor){$\sens_{u,j}$};

    \node[rectangle, draw=lightgray, text=lightgray, minimum height = 0.6cm, minimum width = 1.8cm, node distance=3.7cm, align=left, above=1.6cm of the_sensor.center] (inactive_sensor){$\sens_{u,j'}$};

    \node[below right = 0.4cm and 0.3cm of the_sensor.west] (downvdots) {\scalebox{0.8}{$\vdots$}};
    \node[above right = 0.95cm and 0.3cm of the_sensor.west] (midvdots) {\scalebox{0.8}{$\vdots$}};
    \node[above right = 2.55cm and 0.3cm of the_sensor.west] (upvdots) {\scalebox{0.8}{$\vdots$}};

    \draw [decorate,decoration={brace,amplitude=5pt,mirror,raise=4ex}]
      (1.1,-1.2) -- (7.3,-1.2) node[midway,yshift=-3em]{$K_{\mathrm{a}}$ active sensors out of $\uK$ sensors};
    

    \node[draw, rectangle, minimum width=1.8cm, minimum height=0.6cm, right=1.0cm of the_sensor, align=center] (quant) {$\quant$};
    \node[inner sep=0.8pt, right=0.3cm of quant] (input) {$W_{u,j}$};

    \node[inner sep=0.8pt, right=0.3cm of inactive_sensor, lightgray] (emptyset) {$\emptyset$};
    \draw[->] (inactive_sensor) -- (emptyset);

    \draw[->] (quant) -- (input.west);
    \draw[->] (inactive_sensor.west) ++ (-0.6,0) --  (inactive_sensor.west);
    \draw[->] (the_sensor.west) ++ (-0.6,0) --  (the_sensor.west);
    \draw[->] (the_sensor) -- node[above] {\scriptsize $\vecp_{O_{u,j}}$} (quant.west);
    
    \node[above right = 0.47cm and -0.18cm of the_sensor.west, red] (sensor_label) {active sensor $j$ in zone $u$};

    \node[right = 0cm of emptyset.east, lightgray]  
    (inactive_sensor_label) {sensor $j'$ in zone $u$ is inactive};

    \begin{scope}[on background layer]
        \node[draw, dashed, red, rounded corners, fit={(quant) 
        (input) (the_sensor) (sensor_label)}, inner xsep=0.14cm, inner ysep=0.1cm, yshift=-0.07cm, xshift=-0.0cm] (active_sensor_border) {};
        \node[draw, dashed, lightgray, rounded corners, fit={ (inactive_sensor) (inactive_sensor_label) (emptyset)}, inner xsep=0.2cm, inner ysep=0.2cm, yshift=-0.0cm, xshift=-0.1cm] (inactive_sensor_border) {};
    \end{scope}

\end{tikzpicture}
\vspace{-1cm}
\caption{Block diagram for message generation process in multi-target localization with TUMA.}
\label{fig:msg_gen_mtl_tuma}
\end{figure}

Each active sensor then quantizes the detected target position using a predefined quantization codebook $\setQ = \{\vecq_1, \dots, \vecq_{\uM}\} \subset \setD$. The quantization function is defined as $\quant : \setD \rightarrow [\uM]$ with $\quant(\vecp_{O_{u,j}}) = W_{u,j}$. The index~$W_{u,j}$ is the message transmitted by the active sensor~$(u,j)$ in the TUMA communication framework. A block diagram illustrating this message generation process is shown in Fig.~\ref{fig:msg_gen_mtl_tuma}.

\subsection{Sensing Model} \label{sec:mtl_sensing}

We adopt a probabilistic sensing model based on mmWave 5G SLAM to capture realistic sensing dynamics. Sensing and communication are performed over two different physical links operating over different frequency bands. Specifically, sensing is performed at mmWave, whereas communication occurs over a lower frequency narrowband channel. We assume that the sensors transmit orthogonal frequency division multiplexing~(OFDM) pilot signals to detect targets via reflected signals. 
Using the sensing formulation in~\cite{henk_slam_pdet_2020}, we model the probability that a sensor at position $\vecs$ detects a target at position $\vecp$ as
\begin{equation} \label{eq:pdet}
    p_\text{d}(\vecs,\vecp) = Q_1\mathopen{}\left( 
    \sqrt{\frac{2 \uN_\mathrm{s} \uP_\mathrm{s} \uS \lambda^2}{(4\pi)^3 \uP_\mathrm{n} \lVert \vecs-\vecp\rVert^4}}, \sqrt{\gamma}
    \right).
\end{equation}
Here, $Q_1(\cdot,\cdot)$ is the Marcum-Q function of order $1$, $\uN_\mathrm{s}$ is the number of OFDM symbols used for sensing, referred to as sensing blocklength, and $\uP_\mathrm{s}$ is the per-symbol transmit power allocated for sensing. In addition, $\uS$ corresponds to the radar cross-section of the target, while $\lambda$ denotes the signal wavelength, computed as the speed of light divided by the carrier frequency $f_\mathrm{c}$. The parameter $\uP_\mathrm{n}$ represents the thermal noise power. The threshold $\gamma$ determines the false alarm probability according to $p_\text{fa} = e^{-\gamma/2}$~\cite[Eq.~(8)]{henk_slam_pdet_2020}. Finally, to simplify the analysis, we assume that, while targets can be misdetected, sensors can perfectly localize the detected targets, leveraging the high-resolution capabilities of mmWave sensing~\cite{shahmansoori_2018,abushaban_2018, ge2022slam}.\footnote{The perfect-localization assumption can be relaxed to requiring that the localization error be much smaller than the quantization resolution. Such an error changes the quantized position, and thus the transmitted message, only for targets whose distance from a quantization-region boundary is smaller than the localization error. For uniformly distributed targets, this occurs with probability proportional to the localization error divided by the quantization resolution, which is negligible in this regime. Hence, localization errors have a negligible impact on the reported performance.}

\subsection{Quantization Model}

We employ a fixed uniform quantization over the region~$\setD$, without optimizing the quantization points. Specifically, we assume that the quantized positions define a regular grid within~$\setD$. Each active sensor $(u,j)$ maps the detected target position $\vecp_{O_{u,j}}$ to the nearest point in the quantization codebook $\setQ$ in Euclidean distance:
\begin{equation}
W_{u,j} = \quant\left(\vecp_{O_{u,j}}\right) = \argmin_{m \in [\uM]} \lVert \vecq_m - \vecp_{O_{u,j}} \rVert. \label{eq-scalar_quant}
\end{equation}
In Fig.~\ref{fig:localization}, we illustrate an example of the sensing and quantization process, showing target locations, sensor detections along with their corresponding detection regions, and uniform quantization grids. Note that this figure complements Fig.~\ref{fig:topology}, which shows zone partitioning and AP placements.

\begin{figure}[t!]
    \centering
    \resizebox{0.75\columnwidth}{!}
    {\definecolor{darkgreen}{rgb}{0.0, 0.5, 0.0}
\definecolor{yelloworange}{rgb}{1.0, 0.65, 0.05}
\definecolor{bluecolor}{rgb}{0.0, 0.0, 1.0}
\def\SensorSize{0.25} 

\begin{tikzpicture}

\def\N{4} 
\def\D{4} 
\def\quantCellsize{3} %
\def\APsize{0.3} 
\def\CrossSize{0.1} 
\def\sensradius{\quantCellsize/2.5} 

\draw[black, very thick] (0, 0) rectangle (3*\D, 3*\D);

\foreach \x in {1,2,3} {
    \draw[dotted] (\x*\quantCellsize, 0) -- (\x*\quantCellsize, 3*\D);
}
\foreach \y in {1,2,3} {
    \draw[dotted] (0, \y*\quantCellsize) -- (3*\D, \y*\quantCellsize);
}

\foreach \sx/\sy in {
    1.2/5.73, 1.77/2.43, 5.1/8.3, 10.4/10.0, 8.8/1.6, 10.5/2.6, 4.8/3.8, 5.3/5.2, 6.0/4.7
} {
    \fill[gray!50, opacity=0.3] (\sx, \sy) circle (\sensradius);
}

\draw[dashed, gray] (1.9, 6.23) -- (1.2, 5.73);
\draw[dashed, gray] (1.9, 6.23) -- (1.5, 7.5);
\draw[dashed, gray] (1.2, 5.73) -- (1.5, 7.5);
\draw[dashed, gray] (5.1, 8.3) -- (5.7, 7.5);
\draw[dashed, gray] (5.1, 8.3) -- (4.5, 7.5);
\draw[dashed, gray] (5.7, 7.5) -- (4.5, 7.5);

\draw[dashed, gray] (10.5, 2.6) -- (9.6, 2.2);
\draw[dashed, gray] (10.5, 2.6) -- (10.5, 1.5);
\draw[dashed, gray] (9.6, 2.2) -- (10.5, 1.5);

\draw[dashed, gray] (8.8, 1.6) -- (9.6, 2.2);
\draw[dashed, gray] (8.8, 1.6) -- (10.5, 1.5);
\draw[dashed, gray] (9.6, 2.2) -- (10.5, 1.5);

\draw[dashed, gray] (5.6, 4.3) -- (4.8, 3.8);
\draw[dashed, gray] (5.6, 4.3) -- (4.5, 4.5);
\draw[dashed, gray] (4.8, 3.8) -- (4.5, 4.5);
\draw[dashed, gray] (5.6, 4.3) -- (5.3, 5.2);
\draw[dashed, gray] (5.6, 4.3) -- (4.5, 4.5);
\draw[dashed, gray] (5.3, 5.2) -- (4.5, 4.5);
\draw[dashed, gray] (5.6, 4.3) -- (6.0, 4.7);
\draw[dashed, gray] (5.6, 4.3) -- (4.5, 4.5);
\draw[dashed, gray] (6.0, 4.7) -- (4.5, 4.5);

\foreach \x in {0,1,2,3} {
    \foreach \y in {0,1,2,3} {
        \node[fill=black, scale=1] at (\x*\quantCellsize+1.5, \y*\quantCellsize+1.5) {};
    }
}

\foreach \x/\y in {
    1.2/5.73, 5.1/8.3, 8.8/1.6, 10.5/2.6, 4.8/3.8, 5.3/5.2, 6.0/4.7
} {
    \node[bluecolor, scale=2.5] at (\x, \y) {\textbf{\texttimes}};
}

\foreach \x/\y in {
    1.77/2.43, 10.4/10.0
} {
    \node[lightgray, scale=2.5] at (\x, \y) {\textbf{\texttimes}};
}

\foreach \x/\y in {
    1.9/6.23, 5.7/7.5, 9.6/2.2, 5.6/4.3
} {
    \node[fill=red, circle, scale=1] at (\x, \y) {};
}

\foreach \x/\y in {
    8.9/8.23, 4.8/0.6
} {
    \node[fill=lightgray, circle, scale=1] at (\x, \y) {};
}

\node[scale=2.3] at (1.1,7.9) {$\vecq_{9}$};
\node[scale=2.3] at (4.4,6.9) {$\vecq_{10}$};
\node[scale=2.3] at (11.0,1.0) {$\vecq_4$};
\node[scale=2.3] at (4.0,4.8) {$\vecq_6$};

\node[fill=black, scale=1] at (13.55, 8.8) {};
\node[right, scale=2] at (14, 8.8) {quantized position};
\node[fill=red, circle, scale=1] at (13.55, 7.4) {};
\node[right, scale=2] at (14, 7.4) {detected target};
\node[fill=lightgray, circle, scale=1] at (13.55, 6.0) {};
\node[right, scale=2] at (14, 6.0) {undetected target};
\node[bluecolor, scale=2.5] at (13.55, 4.6) {\textbf{\texttimes}};
\node[right, scale=2] at (14, 4.6) {active sensor};
\node[lightgray, scale=2.5] at (13.55, 3.2) {\textbf{\texttimes}};
\node[right, scale=2] at (14, 3.2) {inactive sensor};


\end{tikzpicture}}
    \caption{An example of multi-target position localization in a square with $\uT=6$ targets, of which $T_\text{d} = 3$ are detected, $\uK=9$ sensors, of which $K_\mathrm{a} = 7$ are active, and $\uM = 16$ quantized positions. Circles around sensors depict their detection range.
    }
    \label{fig:localization}
\end{figure}
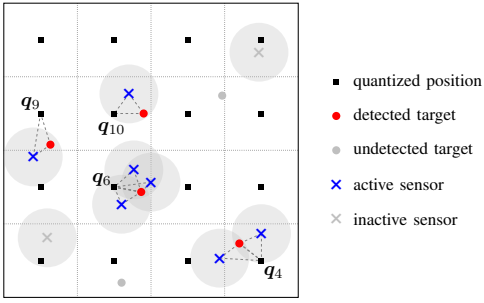

\subsection{Prior Computation}
We next describe the prior that will be used in the Bayesian PME (see Section~\ref{subsec:prior}).

\subsubsection{Sensor Activation Probability}
A sensor becomes active if it detects at least one target. For a sensor at position $\vecs$, the activation probability is given by
\begin{equation}
    p_{\text{active}}(\vecs) = 1 - \frac{1}{|\setD|^\uT} \int_{\setD^\uT} \prod_{i=1}^\uT (1- p_{\text{d}}(\vecs, \vecp_i)) \, \mathrm{d} \vecp_{1:\uT},
\end{equation}
where $\vecp_{1:\uT} = \{\vecp_1,\dots, \vecp_\uT\}$ denote the target positions and $p_d(\cdot,\cdot)$ is the detection probability in \eqref{eq:pdet}. Averaging over all possible sensor locations under the assumed uniform sensor location distribution, we conclude that the sensor activation probability is
\begin{align}
    p_{\text{active}} = \frac{1}{|\setD|} \int_{\setD} p_{\text{active}}(\vecs) \, \mathrm{d} \vecs.
\end{align} 
To obtain a tractable prior model, we approximate the number of active sensors as $K_\mathrm{a} \sim \Bin(\uK,p_{\text{active}})$. This approximation assumes independent sensor activations and neglects the correlation induced by multiple sensors observing the same targets.

\subsubsection{Distribution of Active Sensors in Zones}
Assuming a uniform sensor placement and equal-area zones, we conclude that the number of active sensors in zone \( u \), denoted by \( K_{\mathrm{a},u} \), follows a binomial distribution \( \Bin(K_\mathrm{a}, 1/\uU) \) given \( K_\mathrm{a} \).
Marginalizing over $K_\mathrm{a}$, we obtain
    \begin{align}
        p(K_{\mathrm{a},u}=k) 
        &= \sum_{K_\mathrm{a}=0}^\uK \Bin(k; K_\mathrm{a}, 1/\uU) \Bin(K_\mathrm{a}; \uK, p_{\text{active}}).
\end{align}

\subsubsection{Message Selection Probability}
Each active sensor reports the closest detected target. For a sensor at position $\vecs$, the probability that a target at position $\vecp$ is the closest among the detected ones is
\begin{align}
    p_{\text{closest}}(\vecs, \vecp)  = & \frac{1}{|\setD|^{\uT-1}} \int_{\setD^{\uT-1}} \prod_{\vecp_i'\neq \vecp} (1-p_d(\vecs, \vecp'_i) \notag \\ & 
    \cdot \mathbbm{1} \{\lVert \vecs - \vecp_i' \rVert < \lVert \vecs - \vecp\rVert\}) \, \mathrm{d}\vecp_{1:\uT-1}'.
\end{align} 
Let $\setR_m = \{\vecp \in \setD \sothat \lVert \vecq_m - \vecp \rVert \leq \lVert \vecq_{m'} - \vecp \rVert,  \, \forall m'\in[\uM], \, m'\neq m \}$ denote the Voronoi region associated with
the center point $\vecq_m$.
The probability that an active sensor at position~$\vecs$ in zone $u$  selects the message index $m$ is given by
\begin{align}
    p(m &\mid \text{sensor in zone $u$}) \notag \\ &= \frac{1}{|\setD|} \frac{1}{|\setD_u|} \int_{\setD_u} \int_{\setR_m} p_{\text{d}}(\vecs, \vecp) \, p_{\text{closest}}(\vecs, \vecp) \, \mathrm{d}\vecp \mathrm{d}\vecs . 
\end{align}

\subsubsection{Prior for Local Multiplicities}
Given $K_{\mathrm{a},u}$ active users in zone $u$, the local multiplicities $k_{u,m}$ follow a $\Bin(K_{\mathrm{a},u}, p(m \mid \text{sensor in zone $u$}))$ distribution. By marginalizing over $K_\mathrm{a}$ and $K_{\mathrm{a},u}$, we obtain the prior as
\begin{align}
    p&(k_{u,m}=k) \notag \\ &= \sum_{K_\mathrm{a}=0}^{\uK}\sum_{K_{\mathrm{a},u}=0}^{K_\mathrm{a}} \Bin(k; K_{\mathrm{a},u}, p(m \mid \text{sensor in zone $u$}))  \notag \\ & \qquad \qquad  \: \cdot \Bin(K_{\mathrm{a},u}; K_\mathrm{a}, 1/\uU) \Bin(K_\mathrm{a}; \uK, p_{\text{active}}). \label{eqn-final_prior_expression}
\end{align}

\subsection{Performance Metrics}

We evaluate the performance of the proposed TUMA-based multi-target localization system using metrics that capture the impact of sensing, quantization, and communication. To this end, we consider four complementary metrics:

\subsubsection{TV Distance} The TV distance, introduced in~\eqref{eq:tv_dist_expression}, measures the error in estimating the type. This metric characterizes only communication performance, and does not account for sensing misdetections and quantization distortion. Note that when message collisions are absent, the TV distance reduces to the per-user probability of error~\cite[Def.~1]{polyanskiy2017} commonly used in the UMA literature~\cite{polyanskiy2017, fengler_sparc, ngo_2023_random_user_activity}.

\subsubsection{Wasserstein Distance} 
Let $\vecell=[\ell_1,\dots,\ell_\uT]$, where $\ell_i = \sum_{u=1}^\uU \sum_{j=1}^{K_u} \mathbbm{1} \{O_{u,j}=i\}$ is the number of sensors that detect target $i \in [\uT]$, and define the corresponding type as $\vecomega = \vecell/K_\mathrm{a}$. 
We can then define the empirical distribution associated with the true target positions as $\mu = \sum_{i=1}^{\uT} \omega_i \delta(\vecp_i)$ and its estimate at the receiver as $\hat{\mu} = \sum_{j=1}^{\uM} \hat{t}_j \delta(\vecq_j)$.

To capture both quantization and communication errors, we consider the metric
\begin{equation}
\overline{\mathbb{W}}_p = \mathbb{E}\left[\mathbb{W}_p(\mu, \hat{\mu}) \right],    \label{eq1:was_dist}
\end{equation}
where the expectation is over the random target and sensor positions, small-scale fading, and additive noise. Here, $\mathbb{W}_p(\cdot,\cdot)$ is the $p$-Wasserstein distance~\cite[Chap.~2]{Payre2019_optTransport} defined as
\begin{align}
    \mathbb{W}_p(\mu, \hat{\mu}) = & 
    \inf_{e_{i,j} \in [0,1]} \left(
    \sum_{i=1}^\uT \sum_{j=1}^\uM e_{i,j} \lVert \vecp_i - \vecq_j \rVert_p^p
    \right)^{1/p} \notag \\
    & \quad \text{subject to } \textstyle \sum_{i=1}^\uT e_{i,j} = \hat{t}_j, \, \sum_{j=1}^\uM e_{i,j} = \omega_i. \label{eq-wdist}
\end{align}
This metric captures the mismatch between true and estimated types through an optimal transport formulation. This distance and its generalizations have been widely used in multi-target localization~\cite{Hoffman2004,Schuhmacher2008,Rahmathullah2017}.

\subsubsection{Empirical Misdetection Probability} Let \( \setT_\text{d} = \{ \vecp_i \in \setT : \omega_i > 0 \} \) denote the set of detected targets and let $T_\text{d}=|\setT_d|$. The empirical misdetection probability is defined as
\begin{equation} 
    p_\text{md} = \Exop\left[\frac{|\setT \setminus \setT_\text{d}|}{|\setT|} \right] = 1 - \frac{\mathbb{E}\left[T_\text{d} \right]}{\uT}. \label{eq-pmis}
\end{equation}
This metric isolates the sensing errors by quantifying the average fraction of undetected targets.

\subsubsection{GOSPA-Like Cost} To capture in a single metric the impact of sensing and communication errors, as well as quantization effects, we define the following cost function
\begin{equation}
    d^{(c,p)} = \mathbb{E}\left[\left( 
    \mathbb{W}_p(\mu, \hat{\mu})^p + c^p (1-T_\text{d}/\uT)
    \right)^{1/p} \right], \label{eq:gospa_like_cost}
\end{equation}
which extends the GOSPA metric~\cite{Rahmathullah2017} by replacing its assignment-based localization term in~\cite{Rahmathullah2017} with a Wasserstein distance and incorporating misdetection cost via the \( 1-T_\text{d}/\uT \) term. The parameter $c$ governs the tradeoff between localization errors for detected targets and misdetection probability, while $p$ determines how strongly large localization errors influence the overall cost.

\section{Simulation Results and Discussion} \label{Sec:SimResults}

\subsection{Simulation Setup} \label{SubSec:SimSetup}

We consider a D-MIMO system in which $\uB=40$ APs are deployed to form a $3 \times 3$ square grid layout over a square coverage area, as illustrated in Fig.~\ref{fig:topology}. Each AP is equipped with $\uA=4$ antennas, leading to $\uF=160$ receive antennas in total. Following the setup in \cite{cakmak_2025_journal}, we model the LSFC as $\gamma_b(\vecrho) = 1/(1 + \left( \lVert \vecrho - \vecnu_b \rVert /d_0\right)^{\beta})$, where $\vecnu_b$ denotes the position of AP~$b$, the path-loss exponent is $\beta=3.67$ and the $3 \text{ dB}$ cutoff distance is $d_0=13.57 \text{ m}$. The side length of each zone is set to~$\qty{100}{m}$. The AP placement together with the LSFC model induce a partition of the coverage area into $3 \times 3$ nonoverlapping square zones, resulting in $\uU=9$ zones. As in \cite{cakmak_2025_journal}, we define the received SNR as $\text{SNR}_{\text{rx}} = \text{SNR}_{\text{tx}} / (1 + (\varsigma/d_0)^\alpha)$, where $\varsigma=50 \text{ m}$ denotes the distance between a zone centroid (green dot in Fig.~\ref{fig:topology}) and its closest AP. The codewords~$\{\vecc_{u,m}\}$ are 
independently generated by first sampling each entry independently from a Gaussian distribution~$\mathcal{C}\mathcal{N}(0, 1/\uN_\mathrm{c})$, and then normalizing each codeword such that $\lVert \vecc_{u,m}\rVert_2^2=1$. 
The number of MC samples used in the multisource AMP decoder is set to $\uN_{\text{MC}} = 300$. This value was selected empirically by increasing the number of MC samples until no noticeable improvement in the type estimation performance was observed. The maximum multiplicity considered by the Bayesian denoiser is set to $\uK_{\text{max}}=15$. This value is chosen such that the probability mass beyond the truncation level is negligible under the adopted multiplicity prior. We run the decoder for $\uT_\text{AMP}=10$ iterations. This value was selected based on the observed convergence of the channel estimates. We obtain all reported performance metrics by averaging over $100$ independent MC simulation runs, with independent realizations of user positions, target positions, small-scale fading, and additive noise.\footnote{The simulation code is available at \href{https://github.com/okumuskaan/tuma_mtl}{\nolinkurl{https://github.com/okumuskaan/tuma_mtl}} to support reproducible research.}

We evaluate the performance of the proposed TUMA framework in the multi-target localization setting. A total of $\uT=50$ targets and $\uK=200$ sensors are independently and uniformly placed over the coverage area $\setD$ in each simulation run. All performance results are averaged over $100$ independent simulations. For simulations involving the Wasserstein distance~\eqref{eq-wdist} and the GOSPA-like cost function~\eqref{eq:gospa_like_cost}, we set the exponent $p = 2$, which corresponds to the commonly used Euclidean metric. The weight parameter $c$ is  chosen to ensure that a single misdetection 
is equivalent in terms of cost to a localization error where the distance between true and estimated positions coincides the distance between adjacent quantization points under a $6$-bit resolution. For our $300 \text{ m} \times 300 \text{ m}$ area, this results in $c=300/\sqrt{2^6}=37.5 \text{ m}$.

\begin{table}[t] 
\centering
\caption{Simulation Parameters}
\label{tab:simulation_parameters}
\renewcommand{\arraystretch}{1.15}

\begin{tabular}{m{2cm}|m{3.7cm}|m{1.6cm}}
\toprule
& Parameter & Value \\
\midrule

\multirow{4}{=}{\centering\arraybackslash\shortstack{Topology\\parameters}}
& Number of zones, $\uU$ & 9 \\
& Number of APs, $\uB$ & 40 \\
& Antennas per AP, $\uA$ & 4 \\
& Coverage area & $300\times300$ m$^2$ \\

\midrule

\multirow{8}{=}{\centering\arraybackslash\shortstack{Sensing\\parameters}}
& Number of sensors, $\uK$ & $200$ \\
& Number of targets, $\uT$ & $50$ \\
& Sensing blocklength, $\uN_\mathrm{s}$ & $1000$ \\
& Sensing power, $\uP_\mathrm{s}$ & $10^{-3}$ W \\
& Carrier frequency, $f_c$ & $28$ GHz \\
& Radar cross section, $\uS$ & $10$ $\text{m}^2$  \\
& Thermal noise power, $\uP_\mathrm{n}$ & $10^{-8}$ W \\
& Detection threshold, $\gamma$ & $36.84$ \\

\midrule

\multirow{2}{=}{\centering\arraybackslash\shortstack{Quantization\\parameters}}
& Quantization resolution, $\uJ$ & 10 bits \\
& Number of messages, $\uM$ & $2^\uJ$ \\

\midrule

\multirow{5}{=}{\centering\arraybackslash\shortstack{Communication\\parameters}}
& Communication blocklength, $\uN_\mathrm{c}$ & $1000$ \\
& Communication power, $\uP_\mathrm{c}$ & $10^{-3}$ W \\
& Pathloss exponent, $\beta$ & $3.67$ \\
& Reference distance, $d_0$ & $13.57$ m \\
& Received SNR, $\text{SNR}_\text{rx}$ & $10 \text{ dB}$ \\

\midrule

\multirow{3}{=}{\centering\arraybackslash\shortstack{Decoder\\parameters}}
& Maximum multiplicity, $\uK_{\text{max}}$ & $15$ \\
& Monte Carlo samples, $\uN_{\mathrm{MC}}$ & $300$ \\
& AMP iterations, $\uT_\text{AMP}$ & 10 \\

\midrule

\multirow{2}{=}{\centering\arraybackslash\shortstack{GOSPA-like cost\\parameters}}
& Exponent, $p$ & $2$ \\
& Weight parameter, $c$ & $37.5$ m \\

\bottomrule
\end{tabular}
\end{table}

The sensing parameters in the target detection probability in~\eqref{eq:pdet} are set as follows: the radar cross-section area is $\uS=10 \text{ m}^2$, the carrier frequency is $f_\mathrm{c}=28 \text{ GHz}$, and thermal noise power is $\uP_\mathrm{n}=10^{-8} \text{ Watts}$.\footnote{These parameter values are chosen to reflect a typical mmWave setting. The noise power $\uP_\mathrm{n} = 10^{-8} \text{ W}$ corresponds to a $200 \text{ MHz}$ bandwidth and a noise power spectral density of $-174 \text{ dBm/Hz}$. A radar cross-section of $10 \text{ m}^2$ models medium-sized targets, and $f_\mathrm{c} = 28 \text{ GHz}$ is a commonly chosen carrier frequency in mmWave sensing applications.} The per-symbol sensing power is set equal to the communication power $\uP_\mathrm{s}=\uE_\mathrm{c} / \uN_\mathrm{c} =1 \text{ mW}$. The detection threshold is set to $\gamma = 36.84$, corresponding to a false alarm probability $p_\text{fa}=10^{-8}$. In the simulations, false alarms are ignored. 

The total blocklength \( \uN \) is shared between sensing and communication as \( \uN = \uN_\mathrm{s} + \uN_\mathrm{c} \), according to a time-division resource allocation strategy. Similar time-sharing models have been widely adopted in the literature to characterize tradeoffs under limited resources~\cite{zhang_2022_enabling,liu_2020_jointradarcomm,liu_2022_isac}. Unless stated otherwise, we assume \( \uN = 2000 \), with an equal allocation of resources for sensing and communication, i.e., $ \uN_{\mathrm{s}} = \uN_{\mathrm{c}} = 1000$. 

The simulation parameters used throughout this section are summarized in Table~\ref{tab:simulation_parameters}.

\subsection{Results}

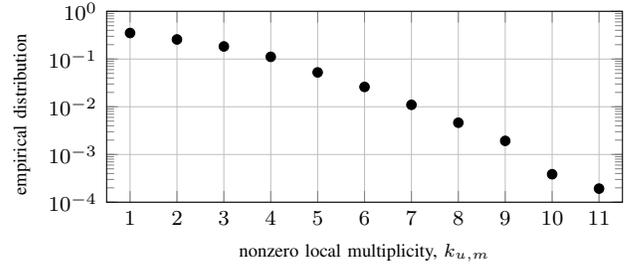
\begin{figure}[t!]
    \centering
    \begin{tikzpicture}
    \tikzstyle{every node}=[font=\footnotesize] 
        \begin{semilogyaxis}[
            scale only axis,
            width=2.7in,
            height=1.0in,
            grid=both,
            grid style={line width=.1pt, draw=gray!10},
            major grid style={line width=.2pt,draw=gray!50},
            xlabel={\scriptsize nonzero local multiplicity, $k_{u,m}$},
            ylabel={\scriptsize empirical distribution},
            ytick = {1e-4, 1e-3, 1e-2, 1e-1,1e0},
            grid=major,
            xtick=data,
            ymin=1e-4,
            ymax=1,
            xmin=0.5,
            xmax=11.5,
            ylabel style = {xshift=-1mm},
        ]
\addplot [
    only marks,
    mark size=1.8pt,
] table[x expr=\coordindex+1, y index=0] {datafile.dat};

    \end{semilogyaxis}
\end{tikzpicture}
    \vspace{-1.0cm}
    \caption{Empirical distribution of nonzero local multiplicities conditioned on message activity ($k_{u,m}>0$) for the multi-target localization scenario described in Section~\ref{SubSec:SimSetup}. Here, we consider $\uT=50$ targets, $\uK=200$ sensors, $\log_2\uM=10$ bits, and a sensing blocklength $\uN_\mathrm{s}=1000$. 
    } 
    \label{fig:multiplicity-emp_dist} 
\end{figure}

We start by verifying that, for the scenario described in Section~\ref{SubSec:SimSetup}, it is likely that sensors within the same zone transmit the same codewords, resulting in local multiplicities $k_{u,m}>1$, which, in turn, justifies the need for the proposed TUMA framework. Notably, as shown in Fig.~\ref{fig:multiplicity-emp_dist}, more than $70\%$ of the transmitted codewords are involved in collisions, i.e., they are transmitted by multiple sensors. 

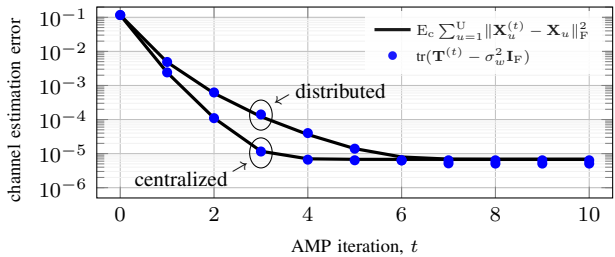
\begin{figure}[t!]
    \centering
    \begin{tikzpicture}
    \tikzstyle{every node}=[font=\footnotesize] 
    \begin{semilogyaxis}[
    scale only axis,
    width=2.7in,
    height=1.0in,
    grid=both,
    grid style={line width=.1pt, draw=gray!10},
    major grid style={line width=.2pt,draw=gray!50},
    xmin=-0.5,
    xmax=10.5,
    xlabel={\scriptsize AMP iteration, $t$},
    ymin=5e-7,
    ymax=0.2,
    ytick = {1e-6, 1e-5, 1e-4,1e-3,1e-2,1e-1,1e0},
    ylabel={\scriptsize channel estimation error},
    ylabel style = {xshift=-1mm},
    legend style={
        draw=none,opacity=.9,
        nodes={scale=0.8, transform shape},  
    },
    legend cell align={left},
    legend image post style={scale=0.8},  
    ]

    \addplot [color=black, mark=none, line width=1.2pt, mark size=1pt]
    table[row sep=crcr]{%
	0	0.11616031495855932 \\
	1	0.002402610958829703 \\
	2	0.00011063577630277331 \\
	3	1.1755417736420586e-05 \\
	4	7.0354516664855214e-06 \\
	5	6.797528996330245e-06 \\
	6	6.782926866467179e-06 \\
	7	6.780049930496581e-06 \\
	8	6.7802444428522125e-06 \\
	9	6.7801637838649895e-06 \\
	10	6.780131248663606e-06 \\
    };
    \addlegendentry{\scriptsize $
    \uE_\mathrm{c} \sum_{u=1}^\uU \mathbb{E} \mathopen{} [ \lVert \mathbf{X}_u^{(t)} - \mathbf{X}_u \rVert_\mathrm{F}^2]$};

    \addplot [color=blue, mark=*, only marks, line width=1.2pt, mark size=1.3pt]
    table[row sep=crcr]{%
	0	0.11638235968563723 \\
	1	0.0024188022235806248 \\
	2	0.00010998861953788625 \\
	3	1.1533530451083638e-05 \\
	4	6.6524798519037885e-06 \\
	5	6.384593349275968e-06 \\
	6	6.3703035623526256e-06 \\
	7	6.369665983659731e-06 \\
	8	6.369621386512883e-06 \\
	9	6.369612577525055e-06 \\
	10	6.369615384312077e-06 \\
    };
    \addlegendentry{\scriptsize $\text{tr}(\mathbf{T}^{(t)} - \sigma_w^2 \mathbf{I}_\uF )$};

    \addplot [color=black, mark=none, line width=1.2pt, mark size=1pt]
    table[row sep=crcr]{%
	0	0.11616031495855929	 \\
	1	0.0048278866771739705	 \\
	2	0.0005896668787201176	 \\
	3	0.00012173810912796916	 \\
	4	3.637201500937405e-05	 \\
	5	1.4019237380646526e-05	 \\
	6	7.973442926724247e-06	 \\
	7	6.90521296146539e-06	 \\
	8	6.847127468391683e-06	 \\
	9	6.84343719821039e-06	 \\
	10	6.843596706121466e-06	 \\
    };

    \addplot [color=blue, mark=*, only marks, line width=1.2pt, mark size=1.3pt]
    table[row sep=crcr]{%
	0	0.11638235968563723	 \\
	1	0.004890919471349241	 \\
	2	0.00062352692316806	 \\
	3	0.0001407462638844404	 \\
	4	4.004521057167275e-05	 \\
	5	1.3957317858864098e-05	 \\
	6	6.323990642757098e-06	 \\
	7	5.181713666853855e-06	 \\
	8	5.10966748907614e-06	 \\
	9	5.109213897023311e-06	 \\
	10	5.109156645599673e-06	 \\
    };

    \coordinate (Center) at (axis cs:3,1.35e-4);
    \coordinate (Center2) at (axis cs:3,1.05e-5);

    \node[align = center] at (axis cs:1.3, 1.8e-6) () {\scriptsize centralized};
    \node[align = center] at (axis cs:4.7,7e-4) () {\scriptsize distributed};

    \node [black, inner sep=0, rotate=300] at (axis cs:2.55,4e-6) {\scriptsize $\uparrow$};

    \node [black, inner sep=0, rotate=120] at (axis cs:3.45,4e-4) {\scriptsize $\uparrow$};

  \end{semilogyaxis} 

    \draw (Center) ellipse [x radius=0.15, y radius=0.2];
    \draw (Center2) ellipse [x radius=0.15, y radius=0.2];

\end{tikzpicture}    
    \vspace{-1.0cm}
    \caption{Channel estimation error vs. multisource AMP iteration for $\text{SNR}_{\text{rx}} = 10 \text{ dB}$, $\uN_\mathrm{c}=\uN_\mathrm{s}=1000$, and ${\uM} = 2^{10}$.}
    \label{fig:channel_est_vs_AMP}
    \vspace{-.3cm}
\end{figure}

To validate the modified multisource AMP decoder, we analyze the convergence behavior of the estimated effective channel matrix $\mathbf{X}$. As shown in Fig.~\ref{fig:channel_est_vs_AMP}, both centralized and distributed AMP decoders exhibit fast convergence, typically within $t=6$ iterations. This aligns with the state evolution result in~\cite[Eq.~(79)]{cakmak_2025_journal}, which predicts that the estimation error~$\uE_\mathrm{c} \sum_{u=1}^\uU \lVert \mathbf{X}_u^{(t)} - \mathbf{X}_u \rVert_\mathrm{F}^2$ converges to  $\text{tr}(\mathbf{T}^{(t)} - \sigma_w^2 \mathbf{I}_\uF )$ in the large-system limit, where $\uN_\mathrm{c} \rightarrow \infty$ with $\uM/\uN_\mathrm{c}$ fixed. 

In Fig.~\ref{fig:channel_est_vs_AMP_imperfectLSFC}, we evaluate the
sensitivity of the proposed centralized decoder to imperfect LSFC
information. Following~\cite{gkiouzepi2024jointmessagedetectionchannel},
we use the centroid of the zone enclosing the sensor rather than its true position to compute the LSFCs and multiplicity priors used by the decoder algorithm. This setup is practically relevant since exact sensor positions and LSFC information are generally unavailable in practice. We observe that LSFC mismatch results in a negligible performance loss, indicating that the proposed decoder is robust. Moreover, Fig.~\ref{fig:channel_est_vs_AMP_imperfectLSFC} illustrates the impact of $\text{SNR}_{\text{rx}}$ on the convergence behavior of multisource AMP. As expected, lower $\text{SNR}_{\text{rx}}$ values lead to higher channel estimation error floors, whereas larger $\text{SNR}_{\text{rx}}$ values yield more accurate estimates.

\begin{figure}[t!]
    \centering
    \begin{tikzpicture}
    \tikzstyle{every node}=[font=\footnotesize] 
    \begin{semilogyaxis}[
    scale only axis,
    width=2.7in,
    height=1.0in,
    grid=both,
    grid style={line width=.1pt, draw=gray!10},
    major grid style={line width=.2pt,draw=gray!50},
    xmin=-0.5,
    xmax=10.5,
    xlabel={\scriptsize AMP iteration, $t$},
    ymin=5e-7,
    ymax=0.2,
    ytick = {1e-6, 1e-5, 1e-4,1e-3,1e-2,1e-1,1e0},
    ylabel={\scriptsize channel estimation error},
    ylabel style = {xshift=-1mm},
    legend style={
        draw=none,opacity=.9,
        nodes={scale=0.8, transform shape},  
    },
    legend cell align={left},
    legend image post style={scale=0.8},  
    ]

    \addplot [color=black, mark=*, line width=1.2pt, mark size=1pt]
    table[row sep=crcr]{%
	0	0.11616031495855932 \\
	1	0.002402610958829703 \\
	2	0.00011063577630277331 \\
	3	1.1755417736420586e-05 \\
	4	7.0354516664855214e-06 \\
	5	6.797528996330245e-06 \\
	6	6.782926866467179e-06 \\
	7	6.780049930496581e-06 \\
	8	6.7802444428522125e-06 \\
	9	6.7801637838649895e-06 \\
	10	6.780131248663606e-06 \\
    };
    \addlegendentry{\scriptsize perfect LSFC information};



    \coordinate (Center) at (axis cs:3,4e-4);
    \coordinate (Center2) at (axis cs:3,1.4e-5);
    \coordinate (Center3) at (axis cs:4,7.0e-5);

    \node[align = center] at (axis cs:1.3, 1.4e-6) () {\scriptsize $\text{SNR}_\text{rx} = 10 \text{ dB}$};
    \node[align = center] at (axis cs:5,1.7e-3) () {\scriptsize $\text{SNR}_\text{rx} = -10 \text{ dB}$};

    \node[align = center] at (axis cs:5.9,2e-5) () {\scriptsize $\text{SNR}_\text{rx} = 0 \text{ dB}$};

    \node [black, inner sep=0, rotate=300] at (axis cs:2.55,4e-6) {\scriptsize $\uparrow$};

    \node [black, inner sep=0, rotate=120] at (axis cs:3.45,8e-4) {\scriptsize $\uparrow$};

    \node [black, inner sep=0, rotate=60] at (axis cs:4.45,3e-5) {\scriptsize $\uparrow$};

    \addplot [color=red, dashed, mark=*, line width=1.2pt, mark size=1pt]
    table[row sep=crcr]{
    0 0.14206235666904798       \\
    1 0.0028611310966103376     \\
    2 0.00019878777954977398    \\
    3 2.649646585442465e-05    \\
    4 8.967196763458853e-06     \\
    5 6.870359973634019e-06      \\
    6 6.582866636958393e-06     \\
    7 6.54413773166847e-06     \\
    8 6.520361212007202e-06     \\
    9 6.50045424028868e-06      \\
    10 6.4904692528730075e-06   \\
    };
    \addlegendentry{\scriptsize imperfect LSFC information};

    \addplot [color=black, mark=*, line width=1.2pt, mark size=1pt]
    table[row sep=crcr]{%
	0	0.1375468196327075	 \\
	1	0.002188064160770371	 \\
	2	0.00046771947413022533	 \\
	3	0.00039940947779938663	 \\
	4	0.00039782546233180975	 \\
	5	0.0003976427612548218	 \\
	6	0.0003976470044133883	 \\
	7	0.0003976470044133883	 \\
	8	0.0003976470044133883	 \\
	9	0.0003976470044133883	 \\
	10	0.0003976474516297258	 \\
    };

    \addplot [color=red, dashed, mark=*, line width=1.2pt, mark size=1pt]
    table[row sep=crcr]{%
	0	0.14216269563753653	 \\
	1	0.002988902522201238	 \\
	2	0.000545398731507387	 \\
	3	0.0004620642041722418	 \\
	4	0.00046015329463114054	 \\
	5	0.00046004415549838294	 \\
	6	0.0004600128054278661	 \\
	7	0.00046001623202410847	 \\
	8	0.00046001445544928345	 \\
	9	0.00046001456068656913	 \\
	10	0.0004600145445673519	 \\
    };

    \addplot [color=black, mark=*, line width=1.2pt, mark size=1pt]
    table[row sep=crcr]{%
0	 0.12626173555770984 	 \\
1	 0.002447198022290817 	 \\
2	 0.00014064003747482097 	 \\
3	 5.661179064261917e-05 	 \\
4	 5.364387184406301e-05 	 \\
5	 5.3533517516285465e-05 	 \\
6	 5.350597626221079e-05 	 \\
7	 5.3501839342653894e-05 	 \\
8	 5.3500166911521184e-05 	 \\
9	 5.3499988886294015e-05 	 \\
10	 5.3499954575679905e-05 	 \\
    };

    \addplot [color=red, dashed, mark=*, line width=1.2pt, mark size=1pt]
    table[row sep=crcr]{%
0	 0.1954642892263693 	 \\
1	 0.0031306063650428863 	 \\
2	 0.0002059308495701151 	 \\
3	 7.199842103609852e-05 	 \\
4	 6.175959051341425e-05 	 \\
5	 6.117273190064762e-05 	 \\
6	 6.1063734573609e-05 	 \\
7	 6.107638971752891e-05 	 \\
8	 6.107072443251831e-05 	 \\
9	 6.107244450067574e-05 	 \\
10	 6.107195101026044e-05 	 \\
    };

  \end{semilogyaxis} 

    \draw (Center) ellipse [x radius=0.15, y radius=0.2];
    \draw (Center2) ellipse [x radius=0.15, y radius=0.2];
    \draw (Center3) ellipse [x radius=0.15, y radius=0.2];

\end{tikzpicture}    
    \vspace{-1.0cm}
    \caption{Channel estimation error $\uE_\mathrm{c} \sum_{u=1}^\uU \mathbb{E}\mathopen{}[ \lVert \mathbf{X}_u^{(t)} - \mathbf{X}_u \rVert_\mathrm{F}^2]$ vs. multisource AMP iteration for centralized decoder with $\uN_\mathrm{c}=\uN_\mathrm{s}=1000$, and ${\uM} = 2^{10}$.}
    \label{fig:channel_est_vs_AMP_imperfectLSFC}
    \vspace{-.3cm}
\end{figure}

\begin{figure}[t!]
    \centering
    \begin{tikzpicture}
    \tikzstyle{every node}=[font=\footnotesize] 
    \begin{semilogyaxis}[
    scale only axis,
    width=2.7in,
    height=1.0in,
    grid=both,
    grid style={line width=.1pt, draw=gray!10},
    major grid style={line width=.2pt,draw=gray!50},
    xmin=-55,
    xmax=15,
    xlabel={\scriptsize received signal to noise ratio, $\text{SNR}_{\text{rx}}$ (dB)},
    ymin=9e-3,
    ymax=1.2,
    ytick = {1e-4,1e-3,1e-2,1e-1,1e0},
    ylabel={\scriptsize average total variation, $\dTV$},
    ylabel style = {xshift=-1mm},
    legend style={draw=none,opacity=.9,at={(0.37,0.01)},anchor=south east},
    ]

    \addplot [color=black, mark=*, line width=1.2pt, mark size=1pt]
    table[row sep=crcr]{%
            -50	 0.9565664096062562	 \\
            -40  0.3461087978446557 \\
            -30  0.1791343064924447 \\
            -20  0.10500101881791768 \\
            -10  0.08067728219009895 \\
            0    0.06531894558324716 \\
            10   0.06526538574898438 \\
    };
    \addlegendentry{\scriptsize centralized};

    \addplot [color=red, mark=*, dashed, line width=1.2pt, mark size=1pt]
    table[row sep=crcr]{%
            -50	 0.9545454545454546	 \\
            -40  0.352126397187386 \\
            -30  0.2169641738911928 \\
            -20  0.13988593776444852 \\ 
            -10  0.11008787045005963 \\
            0    0.10968592629687302 \\
            10   0.10967545554631055 \\ 
    };
    \addlegendentry{\scriptsize distributed};

    \addplot [color=blue, mark=*, dashdotted, line width=1.2pt, mark size=1pt]
    table[row sep=crcr]{%
            -50	 0.99898235 \\
            -40  0.98748343 \\
            -30  0.96733745 \\
            -20  0.93736884 \\
            -10  0.90758478 \\
            0    0.88548567 \\
            10   0.85545438 \\
    };
    \addlegendentry{\scriptsize AMP-DA};
    
  \end{semilogyaxis}
\end{tikzpicture}    
    \vspace{-1.0cm}
    \caption{The average total variation $\dTV$ vs. received signal to noise ratio $\text{SNR}_{\text{rx}}$ for $\uN_\mathrm{s}=\uN_\mathrm{c}=1000$, and ${\uM} = 2^{10}$.}
    \label{fig:tv_vs_snr}
    \vspace{-.3cm}
\end{figure}

In Fig.~\ref{fig:tv_vs_snr}, we evaluate the type estimation performance in terms of $\dTV$ defined in~\eqref{eq:tv_dist_expression} as a function of $\text{SNR}_{\text{rx}}$. As expected, higher $\text{SNR}_{\text{rx}}$ results in lower $\overline{\mathbb{TV}}$ values, and hence, higher estimation accuracy, which is consistent with the channel estimation results in Fig.~\ref{fig:channel_est_vs_AMP_imperfectLSFC}. The centralized decoder performs slightly better, whereas the distributed decoder offers competitive accuracy with better scalability. We further compare our decoders with AMP-DA~\cite{qiao_gunduz_fl}, which let the users pre-equalize the channel to obtain an effective AWGN channel model, and then apply scalar AMP~\cite[Sec.~IV-C]{Meng_NOMA_2021} for type estimation. We emphasize that the pre-equalization step relies on the availability of CSI at the users. Our decoders instead do not require CSI, neither at the users nor the receiver. In our AMP-DA simulations, we normalize the transmitted signal so that the transmit communication power is $1 \text{ mW}$. As shown in Fig.~\ref{fig:tv_vs_snr}, for the received SNR values considered in the figure, AMP-DA results in much higher $\dTV$ than the two TUMA decoders.

\begin{figure}[t!]
\vspace{-3.5mm}
    \centering
    \subfloat[misdetection probability]{
        \begin{tikzpicture}
    \tikzstyle{every node}=[font=\footnotesize] 
    \begin{axis}[
    scale only axis,
    width=1.14in,
    height=1.3in,
    grid=both,
    grid style={line width=.1pt, draw=gray!10},
    major grid style={line width=.2pt,draw=gray!50},
    xmin=0,
    xmax=2000,
    xlabel={\scriptsize sensing blocklength, $\uN_\mathrm{s}$},
    ymin=0,
    ymax=0.8,
    ytick = {0, 0.2, 0.4, 0.6, 0.8},
    yticklabels={
      {\scriptsize 0},
      {\scriptsize 0.2},
      {\scriptsize 0.4},
      {\scriptsize 0.6},
      {\scriptsize \raisebox{-10pt}{0.8}}
    },
    ylabel={\scriptsize average misdetection prob., $p_\text{md}$},
    xtick = {0, 500, 1000, 1500, 2000},
    xticklabels = {0, 500, 1000, 1500, 2000},
    ylabel style = {xshift=-1mm, yshift=-1.5mm},
    xticklabel style = {font=\scriptsize},
    yticklabel style = {font=\scriptsize},
    legend style={draw=none,opacity=.9},
    legend cell align={left}
    ]

    \addplot [color=black, mark=*, line width=1.2pt, mark size=1pt]
    table[row sep=crcr]{%
10.0	0.6977777777777777	 \\ 
100.0	0.36499999999999994	 \\ 
500.0	0.16739130434782604	 \\ 
1000.0	0.12249999999999998	 \\ 
1500.0	0.09894736842105264	 \\ 
1900.0	0.09658536585365855	 \\
    };
  \end{axis}

    \begin{axis}[
        scale only axis,
        width=1.14in,
        height=1.3in,
        xmin=0,
        xmax=2000,
        ymin=0,
        ymax=50,
        axis x line=top,
        axis line style={-},
        axis y line=none,
        xtick={0,500,1000,1500,2000},
        xticklabels={2000,1500,1000,500,0},
        xlabel={\scriptsize communication blocklength, $\uN_\mathrm{c}$},
        xlabel style={yshift=-0.7ex},
        xticklabel style={yshift=-0.7ex, font=\scriptsize}, 
        xtick style={draw=none}, 
    ]
    \end{axis}
\end{tikzpicture}    
        \label{fig:all_metrics_vs_Ns_a}
    }  
    \hspace{-8mm}
    \subfloat[total variation]{
        \begin{tikzpicture}
\tikzstyle{every node}=[font=\footnotesize] 
    \begin{semilogyaxis}[
    scale only axis,
    width=1.14in,
    height=1.3in,
    grid=both,
    grid style={line width=.1pt, draw=gray!10},
    major grid style={line width=.2pt,draw=gray!50},
    xmin=0,
    xmax=2000,
    xtick = {0, 500, 1000, 1500, 2000},
    xticklabels = {0, 500, 1000, 1500, 2000}, 
    xlabel={\scriptsize sensing blocklength, $\uN_\mathrm{s}$},
    ymin=0.004,
    ymax=1,
    ytick = {0.01, 0.1, 1},
    yticklabels={
      {\scriptsize $10^{-2}$},
      {\scriptsize $10^{-1}$},
      {\scriptsize \raisebox{-13pt}{$10^{0}$}}
    },
    ylabel={\scriptsize average TV distance, $\overline{\mathbb{TV}}$},
    ylabel style = {xshift=0mm, yshift=-2.5mm},
    xticklabel style = {font=\scriptsize},
    yticklabel style = {font=\scriptsize},
    legend style={draw=none,opacity=.9,at={(0.50,0.35)},anchor=north},
    legend cell align={left}
    ]

    \addplot [color=black, mark=*, line width=1.2pt, mark size=1pt]
    table[row sep=crcr]{%
10.0	0.005150539748062967	 \\ 
100.0	0.030092580731327017	 \\ 
500.0	0.05666423813963979	 \\ 
1000.0	0.06526538574898438	 \\ 
1500.0	0.06689155755041161	 \\ 
1650.0	0.08194036971742845	 \\ 
1700.0	0.32223302325300807	 \\ 
1800.0	0.8066153634618297	 \\ 
1900.0	0.8396375760856216	 \\
    };  
    \addlegendentry{\scriptsize centralized};

    \addplot [color=red, mark=*, dashed, line width=1.2pt, mark size=1pt]
    table[row sep=crcr]{%
10.0	0.06593797664396768	 \\ 
100.0	0.08681086630961009	 \\ 
500.0	0.10823345756572257	 \\ 
1000.0	0.10970545554631055	 \\ 
1300.0	0.11315504354048421	 \\ 
1400.0	0.17646023692806323	 \\ 
1500.0	0.5476841759414905	 \\ 
1600.0	0.8219543594105182	 \\ 
1700.0	0.8533695851311086	 \\ 
1900.0	0.8853820230961186	 \\
    };  
    \addlegendentry{\scriptsize distributed};
    
  \end{semilogyaxis}
    \begin{axis}[
        scale only axis,
        width=1.14in,
        height=1.3in,
        xmin=0,
        xmax=2000,
        ymin=0,
        ymax=50,
        axis x line=top,
        axis line style={-},
        axis y line=none,
        xtick={0,500,1000,1500,2000},
        xticklabels={2000,1500,1000,500,0},
        xlabel={\scriptsize communication blocklength, $\uN_\mathrm{c}$},
        xlabel style={yshift=-0.7ex},
        xticklabel style={yshift=-0.7ex, font=\scriptsize}, 
        xtick style={draw=none}, 
    ]
    \end{axis}
\end{tikzpicture}    
        \label{fig:all_metrics_vs_Ns_b}
    }
    \vspace{-2mm}
    \subfloat[Wasserstein distance]{
        \begin{tikzpicture}
    \tikzstyle{every node}=[font=\footnotesize] 
        \begin{axis}[
        scale only axis,
        width=1.14in,
        height=1.3in,
        grid=both,
        grid style={line width=.1pt, draw=gray!10},
        major grid style={line width=.2pt,draw=gray!50},
        xtick = {0, 500, 1000, 1500, 2000},
        xticklabels = {0, 500, 1000, 1500, 2000}, 
        xmin=0,
        xmax=2000,
        xlabel={\scriptsize sensing blocklength, $\uN_\mathrm{s}$},
        ymin=0,
        ymax=50,
        ytick = {0, 10, 20, 30, 40, 50},
        yticklabels={
          {\scriptsize 0},
          {\scriptsize 10},
          {\scriptsize 20},
          {\scriptsize 30},
          {\scriptsize 40},
          {\scriptsize \raisebox{-10pt}{50}}
        },
        ylabel={\scriptsize average Wasserstein distance, $\overline{\mathbb{W}}_2$ [m]},
        ylabel style = {xshift=-1mm, yshift=-1.5mm},
        xticklabel style = {font=\scriptsize},
        yticklabel style = {font=\scriptsize},
        legend style={draw=none,opacity=.9,at={(0.4,0.95)},anchor=north},
        legend cell align={left}
        ]

        \addplot [color=black, mark=*, line width=1.2pt, mark size=1pt]
        table[row sep=crcr]{%
10.0	6.304594321475033	 \\ 
100.0	14.04344309851318	 \\ 
500.0	15.96612927680058	 \\ 
1000.0	16.595865722483567	 \\ 
1500.0	16.575856207265	 \\ 
1650.0	16.978633569005115	 \\ 
1700.0	22.589788344305695	 \\ 
1800.0	34.990812018720424	 \\ 
1900.0	35.6374749135196	 \\
        };
        \addlegendentry{\scriptsize centralized};

        \addplot [color=red, mark=*, dashed, line width=1.2pt, mark size=1pt]
        table[row sep=crcr]{%
10.0	27.962479120738584	 \\ 
100.0	23.24578059575684	 \\ 
500.0	22.265049839256132	 \\ 
1000.0	21.963803256638267	 \\ 
1300.0	21.872772411483137	 \\ 
1400.0	23.236448117402376	 \\ 
1500.0	31.4903588995802	 \\ 
1600.0	34.19688443828512	 \\ 
1700.0	35.43485121200646	 \\ 
1900.0	35.85209491914724	 \\
        };
        \addlegendentry{\scriptsize distributed};

        \addplot [color=gray, line width=3pt, mark size=1pt]
        table[row sep=crcr]{%
10	    3.8382072778183045	 \\ 
1990	3.8382072778183045	 \\
        } 
        node [black, pos=.5, inner sep=0, yshift=12pt, align=center] {\scriptsize perfect communication}
        node [black, pos=.5, inner sep=0, yshift=5pt] {\scriptsize $\downarrow$}
        ;
    \end{axis}

    \begin{axis}[
        scale only axis,
        width=1.14in,
        height=1.3in,
        xmin=0,
        xmax=2000,
        ymin=0,
        ymax=50,
        axis x line=top,
        axis line style={-},
        axis y line=none,
        xtick={0,500,1000,1500,2000},
        xticklabels={2000,1500,1000,500,0},
        xlabel={\scriptsize communication blocklength, $\uN_\mathrm{c}$},
        xlabel style={yshift=-0.7ex},
        xticklabel style={yshift=-0.7ex, font=\scriptsize}, 
        xtick style={draw=none}, 
    ]
    \end{axis}
\end{tikzpicture} 
        \label{fig:all_metrics_vs_Ns_c}
    }
    \hspace{-5mm}
    \subfloat[GOSPA-like cost]{
    \begin{tikzpicture}
    \tikzstyle{every node}=[font=\footnotesize] 
        \begin{axis}[
        scale only axis,
        width=1.14in,
        height=1.3in,
        grid=both,
        grid style={line width=.1pt, draw=gray!10},
        major grid style={line width=.2pt,draw=gray!50},
        xmin=0,
        xmax=2000,
        xlabel={\scriptsize sensing blocklength, $\uN_\mathrm{s}$},
        ymin=0,
        ymax=50,
        ytick = {0, 10, 20, 30, 40, 50},
        yticklabels={
          {\scriptsize 0},
          {\scriptsize 10},
          {\scriptsize 20},
          {\scriptsize 30},
          {\scriptsize 40},
          {\scriptsize \raisebox{-10pt}{50}}
        },
        ylabel={\scriptsize average cost, $d^{(37.5 \text{ m},2)}$ [m]},
        ylabel style = {xshift=-1mm, yshift=-1mm},
        xtick = {0, 500, 1000, 1500, 2000},
        xticklabels = {0, 500, 1000, 1500, 2000},
        xticklabel style = {font=\scriptsize},
        yticklabel style = {font=\scriptsize},
        legend style={draw=none,opacity=.9,at={(0.41,0.96)},anchor=north},
        legend cell align={left}
        ]

        \addplot [color=black, mark=*, line width=1.2pt, mark size=1pt]
        table[row sep=crcr]{%
10.0	31.95305790622198	 \\ 
100.0	26.655197317993675	 \\ 
500.0	22.14297418646915	 \\ 
1000.0	21.158647973788518	 \\ 
1500.0	20.34462449508545	 \\ 
1650.0	20.71258626064572	 \\ 
1700.0	25.471830763424318	 \\ 
1800.0	36.85905425095554	 \\ 
1900.0	37.09482746748218	 \\
        };
        \addlegendentry{\scriptsize centralized};

        \addplot [color=red, mark=*, dashed, line width=1.2pt, mark size=1pt]
        table[row sep=crcr]{%
10.0	42.01163126628939	 \\ 
100.0	31.62438762140924	 \\ 
500.0	26.795308206951737	 \\ 
1000.0	25.60949342521858	 \\ 
1300.0	25.583543507577898	 \\ 
1400.0	26.144178079039352	 \\ 
1500.0	33.54454431982002	 \\ 
1600.0	36.10819720348036	 \\ 
1700.0	37.13296084002664	 \\ 
1900.0	37.26197527919144	 \\
        };
        \addlegendentry{\scriptsize distributed};

        \addplot [color=gray, line width=3pt, mark size=1pt]
        table[row sep=crcr]{%
10.0	31.203535524380772	 \\ 
100.0	22.45232360152279	 \\ 
500.0	15.636914181113144	 \\ 
1000.0	13.904157116038976	 \\ 
1500.0	13.112992988158628	 \\
1990.0	12.850883980812764	 \\
}
        node [black, pos=.5, inner sep=0, yshift=-14pt, align=center] {\scriptsize perfect communication}
        node [black, pos=.5, inner sep=0, yshift=-5pt] {\scriptsize $\uparrow$}
        ;

        \addplot [color=black, mark=*, line width=1.2pt, mark size=1pt]
        table[row sep=crcr]{%
10.0	31.95305790622198	 \\ 
100.0	26.655197317993675	 \\ 
500.0	22.14297418646915	 \\ 
1000.0	21.158647973788518	 \\ 
1500.0	20.34462449508545	 \\ 
1650.0	20.71258626064572	 \\ 
1700.0	25.471830763424318	 \\ 
1800.0	36.85905425095554	 \\ 
1900.0	37.09482746748218	 \\
        };

      \end{axis}

    \begin{axis}[
        scale only axis,
        width=1.14in,
        height=1.3in,
        xmin=0,
        xmax=2000,
        ymin=0,
        ymax=50,
        axis x line=top,
        axis line style={-},
        axis y line=none,
        xtick={0,500,1000,1500,2000},
        xticklabels={2000,1500,1000,500,0},
        xlabel={\scriptsize communication blocklength, $\uN_\mathrm{c}$},
        xlabel style={yshift=-0.7ex},
        xticklabel style={yshift=-0.7ex, font=\scriptsize}, 
        xtick style={draw=none}, 
    ]
    \end{axis}
    
    \end{tikzpicture}
    \label{fig:all_metrics_vs_Ns_d}
    }
    \caption{Performance results vs. sensing blocklength $\uN_\mathrm{s}$ for $\uK = 200$ sensors, $\uT = 50$ targets, $\text{SNR}_\text{rx} = 10\dB$, and $\uM = 2^{10}$ messages. The total blocklength is $\uN = 2000$, with communication blocklength $\uN_\mathrm{c} = \uN - \uN_\mathrm{s}$.}
    \label{fig:all_metrics_vs_Ns}
    \vspace{-0.6cm}
\end{figure}
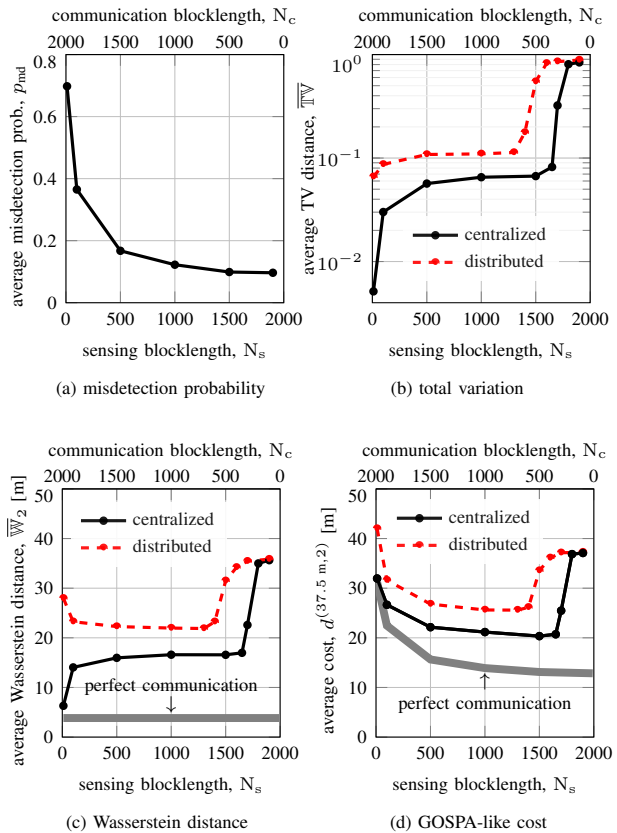

In Fig.~\ref{fig:all_metrics_vs_Ns}, we explore the tradeoff between sensing and communication under a total blocklength of \( \uN = 2000 \), with $\log_2\uM=10$ bits and $\text{SNR}_\text{rx} = 10 \text{ dB}$. We vary the sensing blocklength \( \uN_{\mathrm{s}} \) from $10$ to $1900$ and observe the resulting performance trends. As shown in Fig.~\ref{fig:all_metrics_vs_Ns}(a), increasing the number of sensing symbols \( \uN_{\mathrm{s}} \) improves detection reliability, leading to a lower empirical misdetection probability. However, this comes at the cost of fewer symbols available for communication, which limits the decoder’s ability to distinguish between messages. This degradation is evident in the rising total variation distance~$\overline{\mathbb{TV}}$ in Fig.~\ref{fig:all_metrics_vs_Ns}(b). The Wasserstein distance $\overline{\mathbb{W}}_2$ in Fig.~\ref{fig:all_metrics_vs_Ns}(c), which captures average localization error over detected targets, exhibits a similar trend: improved detection at small \( \uN_{\mathrm{s}} \) is offset by rising communication errors at large values. The GOSPA-like cost in Fig.~\ref{fig:all_metrics_vs_Ns}(d) integrates all three aspects: sensing, quantization, and communication, to provide a unified performance measure. It reveals a clear tradeoff with a minimum in the range \( 1000 \leq \uN_{\mathrm{s}} \leq 1500 \), where the balance between detection accuracy and communication reliability is optimal. This highlights the importance of jointly optimizing sensing and communication resources in practical~integrated~systems.

Finally, in Fig.~\ref{fig:perf_vs_J}, we analyze the impact of the quantization resolution~\(\log_2 \uM \) on communication performance and overall estimation accuracy. As seen in Fig.~\ref{fig:perf_vs_J}(a), the average total variation~$\overline{\mathbb{TV}}$ increases with the number of quantization bits due to the expanded message space, which makes decoding more challenging. In contrast, Fig.~\ref{fig:perf_vs_J}(b) shows that the GOSPA-like cost exhibits a non-monotonic trend, achieving a minimum at an intermediate quantization level around $10$ bits. This reflects a fundamental tradeoff: higher resolution improves target localization but increases communication errors. Note that the performance peak at around $10$ bits cannot be achieved with orthogonal codebooks. Indeed, the total number of codewords~$\overline{\uM} = 9216$ far exceeds the blocklength $\uN_\mathrm{c} = 1000$. 

\begin{figure}[t!]
\vspace{-5mm}
    \centering
    \hspace{-5mm}
    \subfloat[total variation]{
        \begin{tikzpicture}
    \tikzstyle{every node}=[font=\footnotesize] 
        \begin{semilogyaxis}[
        scale only axis,
        width=1.2in,
        height=1.3in,
        grid=both,
        grid style={line width=.1pt, draw=gray!10},
        major grid style={line width=.2pt,draw=gray!50},
        xmin=1,
        xmax=13,
        xlabel={\scriptsize number of quantization bits, $\log_2 \uM$},
        ymin=0.01,
        ymax=1,
        ytick = {0.01, 0.1, 1},
        xtick = {2,4,6,8,10,12},
        ylabel={\scriptsize average total variation, $\overline{\mathbb{TV}}$},
        ylabel style = {xshift=-1mm, yshift=-2mm},
    xticklabel style = {font=\scriptsize},
    yticklabel style = {font=\scriptsize},
tick label style={font=\scriptsize},
        legend style={draw=none,opacity=.9,at={(0.50,0.95)},anchor=north},
        legend cell align={left}
        ]
    
        \addplot [color=black, mark=*, line width=1.2pt, mark size=1pt]
        table[row sep=crcr]{%
2.0	0.02743080485321978	 \\ 
4.0	0.04779191707514203	 \\ 
6.0	0.056381372881989085	 \\ 
8.0	0.058489166393780025	 \\ 
10.0	0.06526538574898438	 \\ 
12.0	0.06807254319065198	 \\
        };
        \addlegendentry{\scriptsize centralized};

        \addplot [color=red, mark=*, dashed, line width=1.2pt, mark size=1pt]
        table[row sep=crcr]{%
2.0	0.0316580844193732	 \\ 
4.0	0.07420790657400433	 \\ 
6.0	0.10392021422438462	 \\ 
8.0	0.10524020848361292	 \\ 
10.0	0.10970545554631055	 \\ 
12.0	0.2983164138387835	 \\
        };
        \addlegendentry{\scriptsize distributed};

      \end{semilogyaxis}
    \end{tikzpicture}
        \label{fig:perf_vs_J_a}
    }
    \hspace{-8mm}
    \subfloat[GOSPA-like cost]{
        \begin{tikzpicture}
    \tikzstyle{every node}=[font=\footnotesize] 
        \begin{axis}[
        scale only axis,
        width=1.2in,
        height=1.3in,
        grid=both,
        grid style={line width=.1pt, draw=gray!10},
        major grid style={line width=.2pt,draw=gray!50},
        xmin=1,
        xmax=13,
        xlabel={\scriptsize number of quantization bits, $\log_2 \uM$},
        ymin=0,
        ymax=50,
        ytick = {0, 10, 20, 30, 40, 50},
        xtick = {2,4,6,8,10,12},
        ylabel={\scriptsize average cost, $d^{(37.5 \text{ m},2)}$ [m]},
        ylabel style = {xshift=-1mm, yshift=-1mm},
    xticklabel style = {font=\scriptsize},
    yticklabel style = {font=\scriptsize},
        legend style={draw=none,opacity=.9,at={(0.65,0.95)},anchor=north},        
        legend cell align={left}
        ]
    
        \addplot [color=black, mark=*, line width=1.2pt, mark size=1pt]
        table[row sep=crcr]{%
2.0	    61.61181686264629	 \\ 
4.0	    35.23545528612795	 \\ 
6.0	    25.186231848872895	 \\ 
8.0	    21.54080449647866	 \\ 
10.0	21.158647973788518	 \\ 
12.0	21.317116458671504	 \\
        };
        \addlegendentry{\scriptsize centralized};

        \addplot [color=red, mark=*, dashed, line width=1.2pt, mark size=1pt]
        table[row sep=crcr]{%
2.0	    62.12813876508836	 \\ 
4.0	    37.03701150497514	 \\ 
6.0	    27.865892111512917	 \\ 
8.0	    25.600966421855734	 \\ 
10.0	25.60949342521858	 \\ 
12.0	29.466260053919634	 \\
        };
        \addlegendentry{\scriptsize distributed};

        \addplot [color=gray, line width=3pt, mark size=1pt]
        table[row sep=crcr]{%
2.0	    61.175506960491455	 \\ 
4.0	    32.7213244856551	 \\ 
6.0	    19.834065417524055	 \\ 
8.0	    15.166672695326483	 \\ 
10.0	13.901334134175995	 \\ 
12.0	12.822158444197157	 \\
        }
        node [black, pos=.5, inner sep=0, yshift=-49pt, xshift=4pt, align=center] {\scriptsize perfect \\ \scriptsize communication}
        node [black, pos=.5, inner sep=0, yshift=-37pt, xshift=15pt, rotate=315] {\scriptsize $\uparrow$}
        ;

        \addplot [color=black, mark=*, line width=1.2pt, mark size=1pt]
        table[row sep=crcr]{%
2.0	    61.61181686264629	 \\ 
4.0	    35.23545528612795	 \\ 
6.0	    25.186231848872895	 \\ 
8.0	    21.54080449647866	 \\ 
10.0	21.158647973788518	 \\ 
12.0	21.317116458671504	 \\
        };

        \addplot [color=red, mark=*, dashed, line width=1.2pt, mark size=1pt]
        table[row sep=crcr]{%
2.0	    62.12813876508836	 \\ 
4.0	    37.03701150497514	 \\ 
6.0	    27.865892111512917	 \\ 
8.0	    25.600966421855734	 \\ 
10.0	25.60949342521858	 \\ 
12.0	29.466260053919634	 \\
        };

      \end{axis}
\end{tikzpicture}
        \label{fig:perf_vs_J_b}
    }
    \caption{Performance results vs. quantization bits $\log_2{\uM}$ for $\uK = 200$ sensors, $\uT = 50$ targets, $\text{SNR}_\text{rx} = 10\dB$, and sensing blocklength $\uN_\mathrm{s} = 1000$.}
    \label{fig:perf_vs_J}
\end{figure}
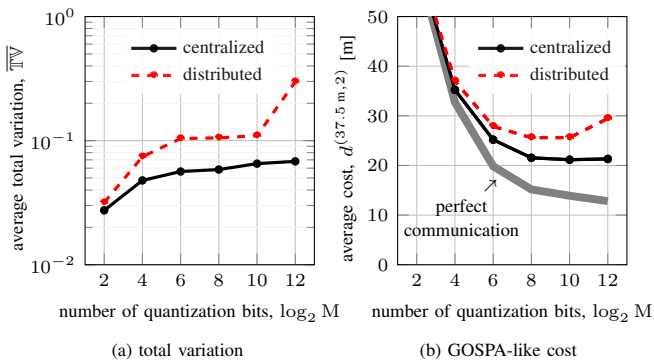

\section{Conclusion} \label{Sec:Conclusion}

We have proposed TUMA for D-MIMO systems---a novel framework that allows estimation of message multiplicities under fading channels without requiring CSI at the users or the receiver. 
A key element of our design is location-based codebook partitioning, where users with similar large-scale fading coefficients select the same codebook. This grouping reduces path-loss variability across users. 
To decode the transmitted message set and their multiplicities, we developed a Bayesian decoder based on multisource approximate message passing. Centralized and distributed implementations were developed, offering a tradeoff between computational complexity (scalability) and performance in large-scale D-MIMO systems. We further demonstrated the use of TUMA in a practical application involving multi-target localization. Our framework includes a probabilistic sensing model, position quantization, and an unsourced communication layer built on the proposed TUMA scheme. We demonstrated through simulations that message collisions are highly prevalent in realistic localization scenarios, reinforcing the need for multiplicity-aware decoding. Our simulation results shed light on the fundamental tradeoffs between sensing, quantization, and communication that arise in the considered multi-target localization scenario. 

The quantization resolutions in the simulations were limited to $12$ bits, since the exponential growth of the message space makes AMP-based decoding computationally infeasible at higher resolutions. Larger information payloads can be accommodated using, for example, the CCS-AMP approach~\cite{amalladine_ccs_amp}, in which larger payloads are split in small sub-blocks, each sub-block is decoded using AMP, and the sub-blocks are stitched together via an outer code. See~\cite{deekshith_tuma_bound} for an application of this approach to TUMA.

\appendix

\subsection{Derivations for the Onsager Correction} \label{app:onsager}

For notational simplicity, as in Section~\ref{subsec:denoiser}, we omit the AMP iteration index~$(t)$ and use $\vecrho_{u,1:k_u}$ and $\vecr_{u}$ to denote $\vecrho_{u,m,1:k_{u,m}}$ and $\vecr_{u,m}$, respectively. The Onsager correction term $\mathbf{Q}_u \in \mathbb{C}^{\uF \times \uF}$ is defined in~\eqref{onsager_expression} as the average Jacobian matrix of the denoiser function $\eta(\cdot)$. Here, for a complex number $r=r_x+jr_y$, the Wirtinger derivative is defined as $\partial (\cdot)/\partial r = \left( \partial (\cdot)/\partial r_x - j \partial (\cdot)/\partial r_y   \right)/2$. 
Using \eqref{eqn:posterior_mults}--\eqref{eqn:mmse_est}, we decompose the denoiser component $[\eta(\vecr_u)]_b$~as
\begin{align} 
[\eta(\vecr_u)]_b &= [\vecr_u]_b \underbrace{\sum_{k=1}^{\uK} \underbrace{\frac{A_b(\vecr_u,k)}{B(\vecr_u,k)}}_{F_b(\vecr_u, k)} \cdot \underbrace{\frac{C(\vecr_u,k)}{D(\vecr_u)}}_{G(\vecr_u, k)}}_{H_b(\vecr_u)},
\end{align} 
where 
\begin{subequations}    
\begin{align} A_b(\vecr_u,k) &= \int_{\setD_u^k} c_{b,u,k} \, p(\vecr_u \mid  \vecrho_{u,1:k}) \, \mathrm{d}\vecrho_{u,1:k}, \\
c_{b,u,k} &= \frac{\sqrt{\uE_\mathrm{c}} [\textstyle \sum_{i=1}^k\matSigma(\vecrho_{u,i})]_{b,b}}{[\mathbf{T}]_{b,b} + \uE_\mathrm{c} [\sum_{i=1}^k\matSigma(\vecrho_{u,i})]_{b,b} }, \label{deriv-eqn-cbp}  \\
B(\vecr_u,k) &= \int_{\setD_u^k} p(\vecr_u \mid  \vecrho_{u,1:k}) \, \mathrm{d}\vecrho_{u,1:k}, \\ 
C(\vecr_u,k) &= \int_{\setD_u^k} p(\vecr_u \mid  \vecrho_{u,1:k}) \, p(k_u = k) \, \mathrm{d}\vecrho_{u,1:k}, \\ D(\vecr_u) &= \sum_{\ell=0}^{\uK} C(\vecr_u, \ell) . \end{align}
\end{subequations}
Note that in (\ref{deriv-eqn-cbp}) we used that $\mathbf{T}$ and $\sum_{i=1}^k \matSigma(\vecrho_{u,i})$ are both diagonal matrices. Then, the derivative of the denoiser is computed as 
\begin{subequations} \label{eqn_deriv30}
    \begin{align} \frac{\partial [\eta(\vecr_u)]_b}{\partial [\vecr_u]_a} &= \frac{\partial}{\partial [\vecr_u]_a} \left( [\vecr_u]_b H_b(\vecr_u) \right) \\ &= \delta(a-b) H_b(\vecr_u) + [\vecr_u]_b \frac{\partial H_b(\vecr_u)}{\partial [\vecr_u]_a},
    \end{align} 
\end{subequations}
where we note that $H_b(\vecr_u) \in \mathbb{R}$. Since $H_b(\vecr_u)=\sum_{k=1}^\uK F_b(\vecr_u,k) \cdot G(\vecr_u,k)$, we have
\begin{align}
    &\frac{\partial H_b(\vecr_u)}{\partial [\vecr_u]_a} \notag \\ & \qquad= \sum_{k=1}^\uK \left( G(\vecr_u, k) \frac{\partial F_b(\vecr_u, k)}{\partial [\vecr_u]_a} + F_b(\vecr_u, k) \frac{\partial G(\vecr_u,k)}{\partial [\vecr_u]_a} \right). \label{eqn_deriv31}
\end{align}
The derivatives of $F_b(\vecr_u,k)=A_b(\vecr_u,k)/B(\vecr_u,k)$ and $G(\vecr_u,k)=C(\vecr_u,k)/D(\vecr_u)$ are given by
\begin{subequations} \label{eqn_deriv32}
    \begin{align} \frac{\partial F_b(\vecr_u,k)}{\partial [\vecr_u]_a} &= \frac{\frac{\partial A_b(\vecr_u,k)}{\partial [\vecr_u]_a} - F_b(\vecr_u,k) \frac{\partial B(\vecr_u,k)}{\partial [\vecr_u]_a}}{B(\vecr_u,k)}, \label{eqn_deriv_Fb}\\ \frac{\partial G(\vecr_u,k)}{\partial [\vecr_u]_a} &= \frac{\frac{\partial C(\vecr_u,k)}{\partial [\vecr_u]_a} - G(\vecr_u,k) \frac{\partial D(\vecr_u)}{\partial [\vecr_u]_a}}{D(\vecr_u)}. \label{eqn_deriv_G}
    \end{align}
\end{subequations}
The derivatives on the RHS of \eqref{eqn_deriv_Fb} and \eqref{eqn_deriv_G} can be evaluated as follows
\begin{subequations} \label{deriv_eqn_allcomps}
    \begin{align} \frac{\partial A_b(\vecr_u,k)}{\partial [\vecr_u]_a} &= \int_{\setD_u^k} c_{b,u,k} \frac{\partial p(\vecr_u \mid \vecrho_{u,1:k})}{\partial [\vecr_u]_a} \, \mathrm{d}\vecrho_{u,1:k}, \\ 
    \frac{\partial B(\vecr_u,k)}{\partial [\vecr_u]_a} &= \int_{\setD_u^k} \frac{\partial p(\vecr_u \mid \vecrho_{u,1:k})}{\partial [\vecr_u]_a} \, \mathrm{d}\vecrho_{u,1:k}, \\ \displaybreak
    \frac{\partial C(\vecr_u,k)}{\partial [\vecr_u]_a} &= \int_{\setD_u^k} p(k_u=k)\, \frac{\partial p(\vecr_u \mid \vecrho_{u,1:k})}{\partial [\vecr_u]_a} \, \mathrm{d}\vecrho_{u,1:k}, \\ 
    \frac{\partial D(\vecr_u)}{\partial [\vecr_u]_a} &= \sum_{\ell=0}^{\uK} \int_{\setD_u^k} p(k_u=\ell) \, \frac{\partial p(\vecr_u \mid \vecrho_{u,1:\ell})}{\partial [\vecr_u]_a} \, \mathrm{d}\vecrho_{u,1:\ell}, \end{align}
\end{subequations}
where we have exchanged the derivative and the integral. This is justified by the Leibniz integral rule, since $p(\vecr_u \mid \vecrho_{u,1:k}) = \mathcal{C}\mathcal{N}(\vecr_u; \veczero, \mathbf{T}+  \uE_\mathrm{c} \sum_{i=1}^k \matSigma(\vecrho_{u,i}))$ is a smooth function of $\vecr_u$ and the integration domain \(\setD_u^k\) is fixed and bounded. The derivative of $p(\vecr_u \mid \vecrho_{u,1:k})$ is then
\begin{align}
    &\frac{\partial p(\vecr_u \mid \vecrho_{u,1:k})}{\partial [\vecr_u]_a} \notag \\
    &= p(\vecr_u \mid \vecrho_{u,1:k}) \frac{\partial}{\partial [\vecr_u]_a} \sum_{f=1}^{\uF} \frac{-([\vecr_u]^2_{f,x} + [\vecr_u]^2_{f,y})}{[\mathbf{T}]_{f,f} +  \uE_\mathrm{c} [\sum_{i=1}^k \matSigma(\vecrho_{u,i})]_{f,f}}. \label{eqn_deriv_expression}
\end{align}
To evaluate the Wirtinger derivative on the RHS of \eqref{eqn_deriv_expression}, note that
\begin{align}
    \frac{\partial \sum_{f=1}^{\uF} \frac{-([\vecr_u]^2_{f,x} + [\vecr_u]^2_{f,y})}{[\mathbf{T} + \uE_\mathrm{c} \sum_{i=1}^k \matSigma(\vecrho_{u,i})]_{f,f}} }{\partial [\vecr_u]_{a,x}}  = \frac{-2[\vecr_u]_{a,x}}{[\mathbf{T} +  \uE_\mathrm{c} \sum_{i=1}^k \matSigma(\vecrho_{u,i})]_{a,a}}, 
\end{align}
and
\begin{align}
    \frac{\partial \sum_{f=1}^{\uF} \frac{-([\vecr_u]^2_{f,x} + [\vecr_u]^2_{f,y})}{[\mathbf{T} +  \uE_\mathrm{c} \sum_{i=1}^k \matSigma(\vecrho_{u,i})]_{f,f}} }{\partial [\vecr_u]_{a,y}} = \frac{-2[\vecr_u]_{a,y}}{[\mathbf{T} +  \uE_\mathrm{c} \sum_{i=1}^k \matSigma(\vecrho_{u,i})]_{a,a}}.
\end{align} 
Hence,
\begin{align}
    \frac{\partial \sum_{f=1}^{\uF} \frac{-([\vecr_u]^2_{f,x} + [\vecr_u]^2_{f,y})}{[\mathbf{T} +  \uE_\mathrm{c} \sum_{i=1}^k \matSigma(\vecrho_{u,i})]_{f,f}} }{\partial [\vecr_u]_{a}} = \frac{-[\vecr_u]^{\ast}_{a}}{[\mathbf{T} +  \uE_\mathrm{c} \sum_{i=1}^k \matSigma(\vecrho_{u,i})]_{a,a}}. \label{eqn_deriv_wirting_first}
\end{align}
Substituting \eqref{eqn_deriv_wirting_first} in \eqref{eqn_deriv_expression}, we have
\begin{align}
    \frac{\partial p(\vecr_u \mid \vecrho_{u,1:k})}{\partial [\vecr_u]_a}
    &= \frac{-[\vecr_u]^{\ast}_{a} \, p(\vecr_u \mid \vecrho_{u,1:k})}{[\mathbf{T} +  \uE_\mathrm{c} \sum_{i=1}^k \matSigma(\vecrho_{u,i})]_{a,a}}. \label{deriv_eq_main_simpler}
\end{align}
Substituting \eqref{deriv_eq_main_simpler} into \eqref{deriv_eqn_allcomps}, we obtain
\begin{subequations}     \label{deriv_local_explicit}
\begin{align}
    &\frac{\partial A_b(\vecr_u,k)}{\partial [\vecr_u]_a} \notag \\
    & = -[\vecr_u]_a^\ast \int_{\setD_u^k} \frac{c_{b,u,k} p(\vecr_u \mid \vecrho_{u,1:k})}{[\mathbf{T} + \uE_\mathrm{c} \sum_{i=1}^k \matSigma(\vecrho_{u,i})]_{a,a}} \, \mathrm{d}\vecrho_{u,1:k},  \\ 
    &\frac{\partial B(\vecr_u,k)}{\partial [\vecr_u]_a} \notag \\
    & = -[\vecr_u]_a^\ast \int_{\setD_u^k} \frac{p(\vecr_u \mid \vecrho_{u, 1:k})}{[\mathbf{T} + \uE_\mathrm{c} \sum_{i=1}^k \matSigma(\vecrho_{u,i})]_{a,a}} \, \mathrm{d}\vecrho_{u,1:k}, \\
    &\frac{\partial C(\vecr_u,k)}{\partial [\vecr_u]_a} \notag \\
    & = -[\vecr_u]_a^\ast \int_{\setD_u^k} \frac{p(k_u=k)\,  p(\vecr_u \mid \vecrho_{u,1:k})}{[\mathbf{T} + \uE_\mathrm{c} \sum_{i=1}^k \matSigma(\vecrho_{u,i})]_{a,a}} \, \mathrm{d}\vecrho_{u,1:k}, \\
    &\frac{\partial D(\vecr_u)}{\partial [\vecr_u]_a} \notag \\
    & = -[\vecr_u]_a^\ast \sum_{\ell=0}^{\uK} \int_{\setD_u^\ell} \frac{p(k_u=\ell)\, p(\vecr_u \mid \vecrho_{u,1:\ell})}{[\mathbf{T} + \uE_\mathrm{c} \sum_{i=1}^k \matSigma(\vecrho_{u,i})]_{a,a}} \, \mathrm{d}\vecrho_{u,1:\ell}.  
\end{align}
\end{subequations}
We obtain the desired result by combining \eqref{eqn_deriv30}--\eqref{eqn_deriv32} and \eqref{deriv_local_explicit}.

\subsection{Distributed Multisource AMP}\label{app:dist_AMP}

In the distributed AMP setup, each AP $b$ processes its local effective received signal $\vecr_{b,u,m} \in \mathbb{C}^{\uA}$ for each codeword $\vecc_{u,m}$. The local likelihood $p_b(\vecr_{b,u,m} \mid \vecrho_{u,1:k})$ is given by
\begin{align}
p_b(\vecr_{b,u,m} \mid \vecrho_{u,1:k}) = \mathcal{CN}(\vecr_{b,u,m}; \veczero, \text{Cov}_b), \label{eqn_dist_likelihood}
\end{align}
where $\text{Cov}_b = \mathbf{T}_b + \uE_\mathrm{c} \sum_{i=1}^k \matSigma_b(\vecrho_{u,i})$ is the local $\uA \times \uA$ covariance matrix. Here, $\mathbf{T}_b$ is the covariance matrix of the local residual noise computed at AP~$b$ and $\matSigma_b(\vecrho)=\gamma_b(\vecrho) \mathbf{I}_{\uA}$ is the LSFC between the user in position~$\vecrho$ and AP~$b$. Let $\text{Cov}= \mathbf{T} + \uE_\mathrm{c} \sum_{i=1}^k \matSigma(\vecrho_{u,i})$ denote the global $\uF\times \uF$ covariance matrix. The key observation is that, since both $\matSigma(\vecrho_{u,i})$ and $\mathbf{T}$ have a block diagonal structure, with the $\{\matSigma_b(\vecrho_{u,i})\}$ and the $\{\mathbf{T}_b\}$ as blocks, then
\begin{equation}
p(\vecr_{u,m} \mid \vecrho_{u,1:k}) = \prod_{b=1}^{\uB} p_b(\vecr_{b,u,m} \mid \vecrho_{u,1:k}).
\end{equation} 
The Onsager correction term in distributed AMP is computed locally at each AP~$b$ as $
[\mathbf{Q}_{u,b}]_{a,c} = (1/\uM) \sum_{m=1}^{{\uM}} \partial [\eta_{u}(\vecr_{b,u,m})]_c/\partial [\vecr_{b,u,m}]_a$, and aggregated at the CPU as
$[\mathbf{Q}_u]_{a,c} = \sum_{b=1}^{\uB} [\mathbf{Q}_{u,b}]_{a,c}.$


\begin{IEEEbiography}[{\includegraphics[width=1in,height=1.25in,clip,keepaspectratio]{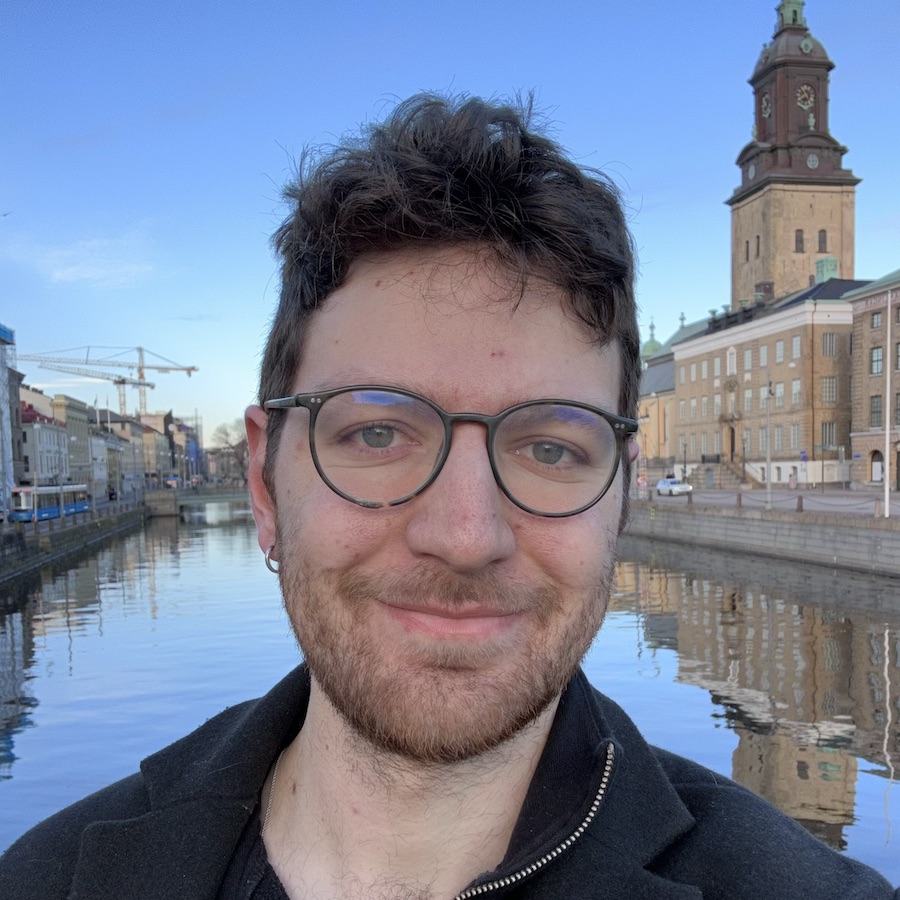}}]{Kaan Okumus}
(Graduate Student Member, IEEE) received the B.Sc. (Hons.) degree in electrical and electronics engineering from Middle East Technical University (METU), Ankara, Türkiye, and the M.Sc. degree in communication systems from École Polytechnique Fédérale de Lausanne (EPFL), Lausanne, Switzerland, in 2023. He is currently pursuing the Ph.D. degree with the Communication Systems Group, Department of Electrical Engineering, Chalmers University of Technology, Gothenburg, Sweden. His research interests include massive random access, wireless communications, and information theory.
\end{IEEEbiography}

\begin{IEEEbiography}[{\includegraphics[width=1in,height=1.25in,clip,keepaspectratio]{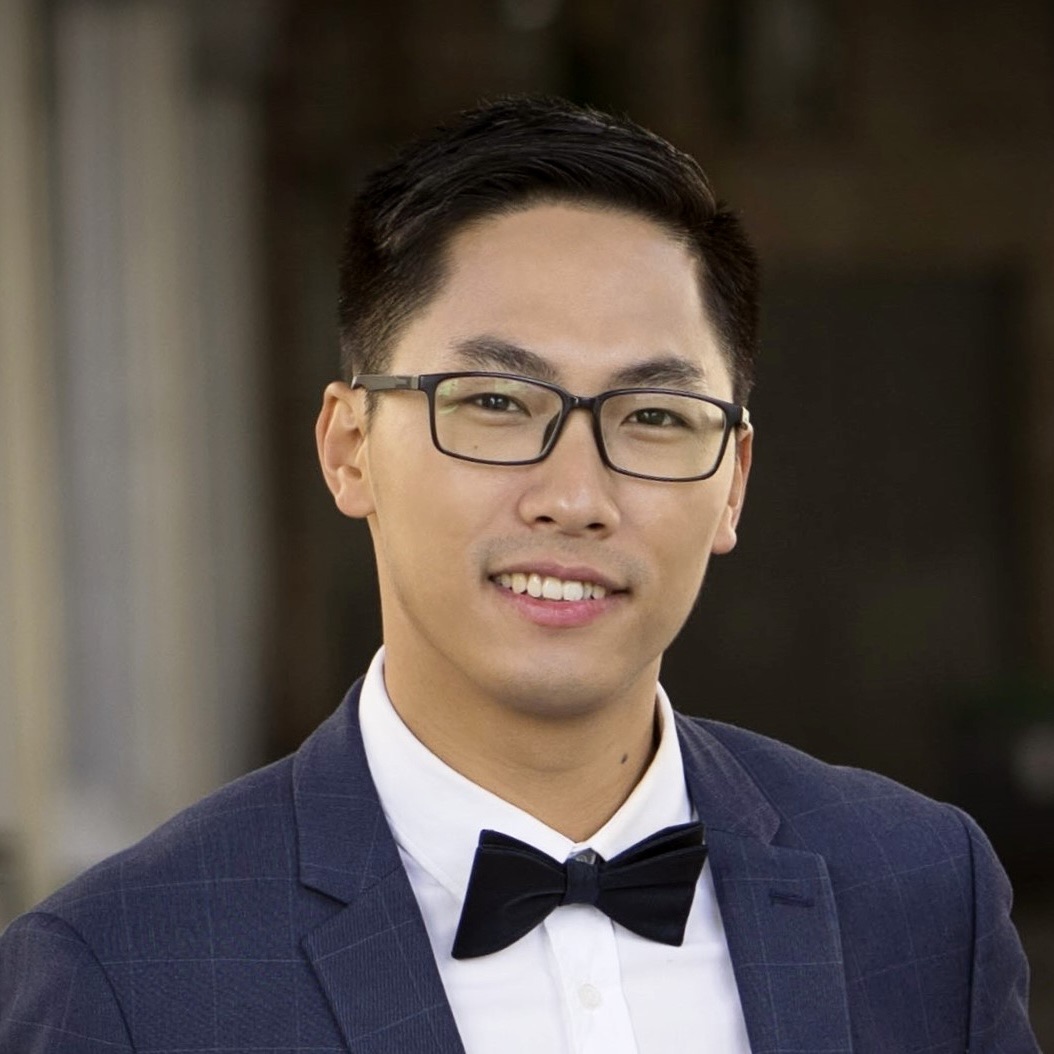}}]{Khac-Hoang Ngo}
(Member, IEEE) received the B.E. degree (Hons.) in electronics and telecommunications from the University of Engineering and Technology, Vietnam National University, Hanoi, Vietnam, in 2014, and the M.Sc. (Hons.) and Ph.D. degrees in wireless communications from the CentraleSupélec, Paris-Saclay University, France, in 2016 and 2020, respectively. His Ph.D. thesis was also realized with Paris Research Center, Huawei Technologies, France. Since September 2024, he has been an Assistant Professor with Linköping University, Sweden. From 2020 to 2024, he was a Post-Doctoral Researcher with the Chalmers University of Technology, Sweden. His research interests include wireless communications and information theory, with an emphasis on massive random access, privacy and security of machine learning, information freshness, MIMO, and non-coherent communications. He is serving as an Associate Editor of IEEE Wireless Communications Letters, and the Chair of the IEEE Information Theory Student \& Outreach Subcommittee. He received the Signal, Image and Vision Ph.D. Thesis Prize by Club EEA, GRETSI, and GdR-ISIS, France, in 2021.
\end{IEEEbiography}

\vspace{-0.5cm}

\begin{IEEEbiography}[{\includegraphics[width=1in,height=1.25in,clip,keepaspectratio]{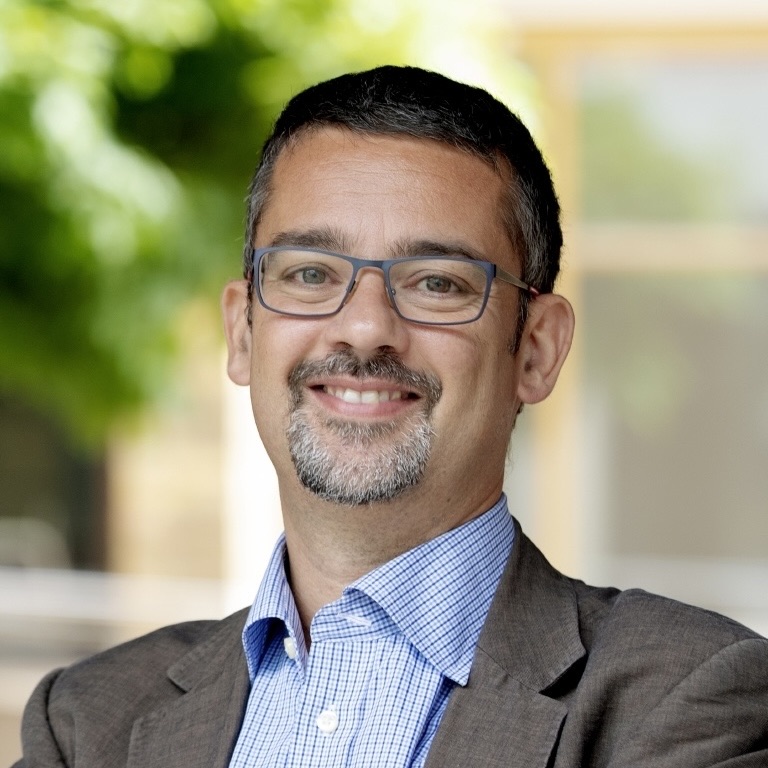}}]{Giuseppe Durisi}
(Senior Member, IEEE) received the Laurea (summa cum laude) and Ph.D. degrees from the Politecnico di Torino, Italy, in 2001 and 2006, respectively. From 2002 to 2006, he was with the Istituto Superiore Mario Boella, Turin, Italy. From 2006 to 2010, he was a Post-Doctoral
Researcher with ETH Zurich, Zürich, Switzerland.
In 2010, he joined the Chalmers University of Technology, Gothenburg, Sweden, where he is currently a
Professor with the Communication Systems Group.
His research interests are in the areas of communication and information theory and machine learning. He was a recipient
of the 2013 IEEE ComSoc Best Young Researcher Award for the Europe,
Middle East, and Africa region, and is the co-author of a paper that won the
Student Paper Award at the 2012 International Symposium on Information
Theory and of a paper that won the 2013 IEEE Sweden VT-COM-IT Joint
Chapter Best Student Conference Paper Award. From 2015 to 2021, he served
as an Associate Editor for \textsc{IEEE Transactions on Communications}.
Since 2024, he has been an Associate Editor for \textsc{IEEE Transactions on
Information Theory}.
\end{IEEEbiography}

\vspace{-0.5cm}

\begin{IEEEbiography}[{\includegraphics[width=1in,height=1.25in,clip,keepaspectratio]{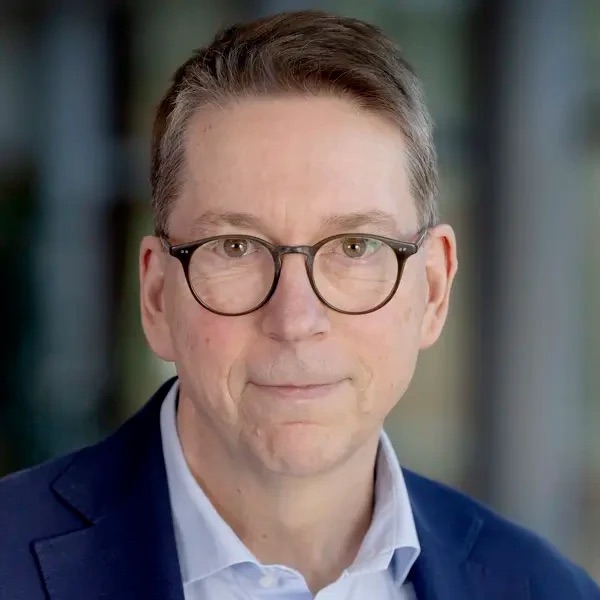}}]{Erik G. Ström}
(Fellow, IEEE) received the M.S.
degree in electrical engineering from the Royal
Institute of Technology (KTH), Stockholm, Sweden,
in 1990, and the Ph.D. degree in electrical engineering from the University of Florida, Gainesville, in
1994. He accepted a post-doctoral position at the
Department of Signals, Sensors, and Systems, KTH,
in 1995. In February 1996, he was appointed as
an Assistant Professor at KTH, and in June 1996,
he joined the Chalmers University of Technology,
Gothenburg, Sweden, where he has been a Professor
of communication systems since June 2003. He currently heads the Division of
Communications, Antennas, and Optical Networks and is the Director of
Chalmers’ Area-of-Advance Information and Communication Technology.
Since 1990, he has acted as a consultant for the Educational Group for Individual Development, Stockholm. His research interests include signal processing
and communication theory in general, and constellation labelings, channel
estimation, synchronization, multiple access, medium access, multiuser detection, wireless positioning, and vehicular communications in particular. He was
a member of the Board of the IEEE VT/COM Swedish Chapter 2000–2006.
He received the Chalmers Pedagogical Prize in 1998, the Chalmers Ph.D.
Supervisor of the Year Award in 2009, and the Chalmers Area of Advance
Award in 2020. He has been a Senior Editor of \textsc{IEEE Transactions on Intelligent Transportation Systems}, a contributing author and associate
editor for Roy. Admiralty Publishers FesGas-series, and was a Co-Guest Editor
for \textsc{Proceedings of the IEEE} special issue on Vehicular Communications
in 2011 and \textsc{IEEE Journal on Selected Areas in Communications}
special issues on Signal Synchronization in Digital Transmission Systems in
2001 and on Multiuser Detection for Advanced Communication Systems and
Networks in 2008.
\end{IEEEbiography}

\end{document}